%% file: main.tex
\documentclass[12pt]{iopart}

\usepackage{amsthm}
\usepackage{amsfonts}
\usepackage{siunitx}
\usepackage{graphicx}
\usepackage{tikz}
\usetikzlibrary{patterns}
\usepackage{multirow}
\usepackage{hyperref}
\usepackage{subcaption}
\usepackage{algorithm, algpseudocode}
\usepackage{siunitx}
\usepackage{booktabs}

\newcommand{\R}{\mathbb{R}}

\newcommand{\loss}{\mathcal{L}}
\newcommand{\disc}{\mathcal{S}}
\newcommand{\D}{\mathcal{D}}
\newcommand{\T}{\mathcal{T}}

\newcommand{\K}{\mathcal{K}}
\newcommand{\F}{\mathcal{F}}

\newcommand\sqnorm[1]{\Vert #1 \Vert^2}
\newcommand\nsqnorm[1]{\Vert #1 \Vert}

\newcommand{\argmin}{\mathop{\mathrm{arg\,min}}}

\newtheorem{lemma}{Lemma}
\newtheorem{assumption}{Assumption}

\newtheorem{definition}{Definition}
\newtheorem{remark}{Remark}

\begin{document}

\title[]{Computed Tomography Reconstruction Using Deep Image Prior and Learned Reconstruction Methods}

\author{Daniel Otero Baguer, Johannes Leuschner, Maximilian Schmidt}

\address{Center for Industrial Mathematics (ZeTeM), University of Bremen, Bibliothekstra{\ss}e 5, 28359 Bremen, Germany}
\ead{\{otero, jleuschn, schmidt4\}@uni-bremen.de}
\vspace{10pt}
\begin{indented}
\item[]February 2020
\end{indented}

\begin{abstract}

In this work, we investigate the application of deep learning methods for computed tomography in the context of having a low-data regime. As motivation, we review some of the existing approaches and obtain quantitative results after training them with different amounts of data. We find that the learned primal-dual has an outstanding performance in terms of reconstruction quality and data efficiency. However, in general, end-to-end learned methods have two issues: a) lack of classical guarantees in inverse problems and b) lack of generalization when not trained with enough data. To overcome these issues, we bring in the deep image prior approach in combination with classical regularization. The proposed methods improve the state-of-the-art results in the low data-regime.

\end{abstract}

\section{Introduction}

Deep learning approaches for solving ill-posed inverse problems currently achieve state-of-the-art reconstruction quality in terms of quantitative results. However, they require large amounts of training data, i.e., pairs of ground truths and measurements, and it is not clear how much is necessary to be able to generalize well. For ill-posed inverse problems arising in medical imaging, such as Magnetic Resonance Imaging (MRI), guided Positron Emission Tomography (PET), Magnetic Particle Imaging (MPI), or Computed Tomography (CT), obtaining such high amounts of training data is challenging. In particular ground truth data is difficult to obtain, as, for example, it is of course impossible to take a photograph of the inside of the body. What learned methods usually consider as ground truths are simulations or high-dose reconstructions obtained with classical methods, such as filtered back-projection (FBP), which work considerably well in the presence of a sufficiently large amount of low-noise measurements. In MRI, it is well possible to obtain those reconstructions, but it requires much time for the acquisition process. Therefore a potential of learned approaches in MRI is to reduce the acquisition times \cite{zbontar2018fastmri}. In other applications such as CT, it would be necessary to expose patients to high doses of X-ray radiation to obtain the required training ground truths.

There is yet another approach called Deep Image Prior (DIP) \cite{ulyanov2018dip} that also uses deep neural networks, for example, a U-Net. However, there is a remarkable difference: it does not need any learning, i.e., the weights of the network are not trained. This approach seems to have low applicability because it requires much time to reconstruct in contrast to learned methods. In the applications initially considered, for example, inpainting, denoising, and super-resolution, it is much easier to obtain or simulate data, which allows for the use of learned methods, and the DIP does not seem to have an advantage. However, these applications are not ill-posed inverse problems in the sense of Nashed \cite{nashed1987ill_posed}. The main issue is that, in some cases, they do have a non-trivial null space, which makes the solution not unique.

In this work, we aim to explore the application of the DIP together with other deep learning methods for obtaining CT reconstructions in the context of having a rather low-data regime. The structure of the paper and the main contributions are organized as follows. In Section~\ref{sec:ct}, we briefly describe the CT reconstruction problem. Section~\ref{sec:review} provides a summary of related articles and approaches, together with some background and insights that we use as motivation. The experienced reader may skip Sections \ref{sec:ct} and \ref{sec:review} and go directly to Section~\ref{sec:regularization}, where we introduce the combination of the DIP with classical regularization methods and obtain theoretical guarantees. Following, in Section~\ref{sec:dip_initial}, we propose a similar approach to the DIP but using an initial reconstruction given by any end-to-end learned method. Finally, in Section~\ref{sec:benchmark}, we present a benchmark of the analyzed methods using different amounts of data from two standard datasets.

\section{Computed Tomography}
\label{sec:ct}

Computed tomography (CT) is one of the most valuable technologies in modern medical imaging \cite{buzug2008computed_tomography}. It allows for a non-invasive acquisition of the inside of the human body using X-rays. Since the introduction of CT in the 1970s, technical innovations like new scan geometries pushed the limits on speed and resolution. Current research focuses on reducing the potentially harmful radiation a patient is exposed to during the scan \cite{buzug2008computed_tomography}. These include measurements with lower intensity or at fewer angles. Both approaches introduce particular challenges for reconstruction methods, that can severely reduce the image quality. In our work, we compare multiple reconstruction methods on these low-dose scenarios for a basic 2D parallel beam geometry (cf.\ Figure \ref{fig:parallel_beam}).

In this case, the forward operator is given by the 2D Radon transform \cite{radon1986radon_trafo} and models the attenuation of the X-ray when passing through a body. We can parameterize the path of an X-ray beam by the distance from the origin $s \in \R$ and angle $\varphi \in [0,\pi]$
\begin{equation}
    L_{s,\varphi}(t) = s \omega\left(\varphi\right) + t \omega^\perp\left(\varphi\right), \quad \omega\left(\varphi\right) := [\cos(\varphi),\, \sin(\varphi)]^T.
\end{equation}
The Radon transform then calculates the integral along the line for parameters $s$ and $\varphi$
\begin{equation}
    A x(s,\varphi) = \int_{\mathbb{R}} x\left(L_{s,\varphi}(t) \right) \, \mathrm{d}t.
\end{equation}
According to Beer-Lambert's law, the result is the logarithm of the ratio between the intensity $I_0$ at the X-ray source and $I_1$ at the detector
\begin{equation}
    A x(s,\varphi) = -\ln \left(\frac{I_1\left(s, \varphi\right)}{I_0\left(s, \varphi\right)}\right) = y\left(s, \varphi\right).
\end{equation}
Calculating the transform for all pairs $(s, \varphi)$ results in a so-called \textit{sinogram}, which we also call observation. To get a reconstruction $\hat{x}$ from the sinogram, we have to invert the forward model. Since the Radon transform is linear and compact, the inverse problem is \textit{ill-posed} in the sense of Nashed \cite{nashed1987ill_posed, natterer2001math_tomography}.

\begin{figure}[htb]
    \centering
    \scalebox{0.7}{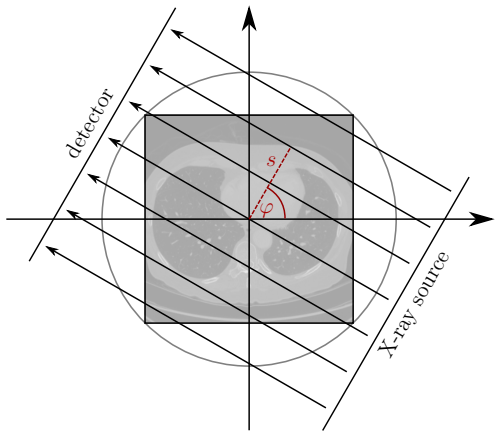}
    \caption{Parallel beam geometry}
    \label{fig:parallel_beam}
\end{figure}

\section{Related approaches and motivation}
\label{sec:review}
In this section, we first review and describe some of the existing data-driven and classical methods for solving ill-posed inverse problems, which have also been applied to obtain CT reconstructions. Following, we review the DIP approach and related works.

In inverse problems one aims at obtaining an unknown quantity, in this case the scan of the human body, from indirect measurements that frequently contain noise \cite{engl, louis89, rieder}. The problem is modeled by an operator $A: X \rightarrow Y$ between Banach or Hilbert spaces $X$ and $Y$ and the measured noisy data or observation
\begin{equation}
y^\delta = A x^\dagger + \tau.
\label{eq:inv_problem}
\end{equation}
The aim is to obtain an approximation $\hat{x}$ for $x^\dagger$ (the true solution), where $\tau$, with $\|\tau\| \le \delta,$ describes the noise in the measurement.

Classical approaches for inverse problems include linear pseudo inverses given by filter functions \cite{louis89} or  non-linear regularized inverses given by the variational approach
\begin{equation}
\label{eq:variational}
\T_\alpha(y^\delta) \in \arg\min_{x\in \D} \disc(Ax,\, y^\delta) + \alpha J(x),
\end{equation}
where $\disc: Y\times Y \to \R$ is the data discrepancy, $J:X \to \R \cup \{\infty\}$ is the regularizer, $\D:=\D(A) \cap \D(J)$ and $\D(A)$, $\D(J)$ are the domains of $A$ and $J$ respectively. Examples of hand-crafted regularizers/priors are $\Vert x \Vert^2$, $\Vert x \Vert_1$ and Total Variation (TV). The value of the regularization parameter $\alpha$ should be carefully selected. One way to do that, in the presence of a validation dataset with some ground truth and observation pairs, is to do a line-search and select the $\alpha$ that yields the best performance on average, assuming there is a uniform noise level. Given validation data $\{x_i^\dagger,\ y_i^\delta\}_{i=1}^N$, the data-driven parameter choice would be
\begin{equation}
\hat{\alpha} := \argmin_{\alpha \in \mathbb{R}_+} {\sum_{i=1}^N{\ell(\T_\alpha(y_i^\delta),\, x_i^\dagger)}},
\end{equation}
where $\ell: X \times X \to \R$ is some similarity measure, such as PSNR or SSIM.

Data-driven regularized inverses for solving inverse problems in imaging have recently had great success in terms of reconstruction quality \cite{adler2018learned, arridge2019, hauptmann2018model}. Three main classes are: end-to-end learned methods \cite{jonas2017iterative, adler2018learned, hauptmann2018model, unser2017, Schwab2019}, learned regularizers \cite{haltmeier18, lunz2018adversarial} and generative networks \cite{bora2017compressed}. For this study, we only focus on the end-to-end learned methods.

\subsection{End-to-end learned methods}
\label{sec:learned_methods}

In this section, we briefly review the most successful end-to-end learned methods. Most of them were implemented and included in our benchmark.

\subsubsection{Post-processing}
This method aims at improving the quality of the filtered back-projection (FBP) reconstructions from noisy or few measurements by applying learned post-processing. Recent works \cite{chen2017convnet_ct, jin2017cnn_imaging, yang2018gan_ct} have successfully used a convolutional neural network (CNN), such as the U-Net \cite{ronneberger2015unet}, to remove artifacts from FBP reconstructions.
In mathematical terms, given a possibly regularized FBP operator $\T_\mathrm{FBP}$, the reconstruction is computed using a network $D_\theta : X \to X$ as
\begin{equation}
 \hat x := [D_\theta \circ \T_\mathrm{FBP}] (y^\delta)
\end{equation}
with parameters $\theta$ of the network that are learned from data.

\subsubsection{Fully learned} 

Related methods aim at directly learning the inversion process from data while keeping the network architecture as general as possible. This idea was successfully applied in MRI by the AUTOMAP architecture \cite{Zhu2018}. The main building blocks consist of fully connected layers. Depending on the problem, the number of parameters can grow quickly with the data dimension. For mapping from sinogram to reconstruction in the LoDoPaB-CT Dataset, such a layer would have over $1000 \times 513 \times 362^2 \approx 67 \cdot 10^9$ parameters. This makes the naive approach infeasible for large CT data.

He \etal \cite{he2020iradon} introduced an adapted two-part network, called iRadonMap. The first part reproduces the structure of the FBP. A fully connected layer is applied along $s$ and shared over the rotation angle dimension $\varphi$, playing the role of the filtering. For each reconstruction pixel $(i,j)$ only sinogram values on the sinusoid $s = i\cos(\varphi) + j\sin(\varphi)$ have to be considered and are multiplied by learned weights. For the example above, the number of parameters in this layer reduces to $513^2 + (362)^2\cdot1000 \approx 130 \cdot 10^6$. The second part consists of a post-processing network. We choose the U-Net architecture for our experiments, which allows for a direct comparison with the FBP + U-Net approach.

\subsubsection{Learned iterative schemes}

Another series of works \cite{jonas2017iterative, adler2018learned, hauptmann2018model} use CNNs to improve iterative schemes commonly used in inverse problems for solving (\ref{eq:variational}), such as gradient descent, proximal gradient descent or hybrid primal-dual algorithms.
The idea is to unroll these schemes with a small number of iterations and replace some operators by CNNs with parameters that are trained using ground truth and observation data pairs.
The simplest one is probably the proximal gradient descent, whose standard version is given by the iteration
\begin{equation}
 x^{(k+1)} = \phi_{J,\,\alpha,\,\lambda_k}(x^{(k)} - \lambda_k A^*(Ax^{(k)}- y^\delta)),
\end{equation}
for $k=0$ to $L-1$, where $\phi_{J,\,\alpha,\,\lambda}:X \to X$ is the proximal operator. The corresponding learned iterative scheme is a partially learned approach, where each iteration is performed by a convolutional network $\psi_{\theta_k}$ that includes the gradients of the data discrepancy and of the regularizer as input in each iteration. Moreover, the number of iterations is fixed and small, e.g.,\ $L=10$. The reconstruction operator is given by $\T_\theta: Y \to X$ with $\T_\theta(y^\delta) = x^{(L)}$ and
\begin{eqnarray*}
x^{(k+1)} &=\psi_{\theta_k}(x^{(k)},\ A^*(Ax^{(k)}- y^\delta), \nabla J(x^{(k)}))\\
x^{(0)} &= A^+(y^\delta)
\end{eqnarray*}
for any pseudo inverse $A^+$ of the operator $A$ and $\theta = (\theta_0, \dots, \theta_{L-1})$. Alternatively, $x^{(0)}$ could be just randomly initialized.

Similarly, more sophisticated algorithms, such as hybrid primal-dual algorithms, can be unrolled and trained in the same fashion. In this work, we used an implementation of the learned gradient descent \cite{jonas2017iterative} and the learned primal-dual method \cite{adler2018learned}.

\bigskip

The above mentioned approaches all rely on a parameterized operator $\T_\theta: Y \to X$, whose parameters $\theta$ are optimized using a training set of $N$ ground truth samples $x^{\dagger}_i$ and their corresponding noisy observations $y^{\delta}_i$. Usually, the empirical mean squared error is minimized, i.e.,
\begin{equation}
\hat{\theta} \in \argmin_{\theta \in \Theta} \frac{1}{N}\sum_{i=1}^N \sqnorm{\T_\theta(y^{\delta}_i) - x^{\dagger}_i}.
\end{equation}

After training, the reconstruction $\hat{x} \in X$ from a noisy observation $y^\delta \in Y$ is given by $\hat{x} = \T_{\hat{\theta}}(y^\delta)$. The main disadvantage of these approaches is that they do not enforce data consistency. As a consequence, some information in the observation could be ignored, yielding a result that might lack important features of the image. In medical imaging, this is critical since it might remove an indication of a lesion.

\subsubsection{Null Space Network}
In order to overcome this issue, in \cite{Schwab2019} the authors introduce a new approach called Null Space Network. It takes the form
\begin{equation}
\F_\theta := \textrm{Id}_X + (\textrm{Id}_X-A^+A)\Psi_\theta,
\end{equation}
where the function $\Psi_\theta:\ X \to X$ is defined by a neural network, $A^+$ is the pseudo inverse of $A$ and $\textrm{Id}_X-A^+A=P_{\mathrm{ker(A)}}$ is the projection onto the null space $\mathrm{ker}(A)$ of $A$.
Consequently, the null space network $\F_\theta$ satisfies the property $A\F_\theta(x) = Ax$ for all $x \in X$. When combined with the pseudo inverse $\T_\theta = \F_\theta \circ A^+$, this yields an end-to-end learned approach with data consistency. Theoretical results for this approach have been proved in \cite{Schwab2019}. We did not include this approach in the comparison, but leave it for a future study.

\subsection{Deep Image Prior}
\label{sec:dip}

The DIP is similar to the generative networks approach and the variational method. However, instead of having a regularization term $J(x)$, the regularization is incorporated by the reparametrization $x = \varphi(\theta,\, z)$, where $\varphi$ is a deep generative network with weights $\theta \in \Theta$, and $z$ is a fixed input, for example, random white noise. The approach is depicted in Figure~\ref{fig:dip-method} and consist in solving
\begin{equation}
\hat{\theta} \in \argmin_{\theta \in \Theta} \sqnorm{A\varphi(\theta,\, z) - y^\delta}, \ \ \ \ \ \hat{x}:=\varphi(\hat{\theta},\, z)\ .
\label{eq:dip_formulation}
\end{equation}
In the original method, the authors use gradient descent with early stopping to avoid reproducing noise. This is necessary due to the overparameterization of the network, which makes it able to reproduce the noise. The regularization is a combination of early stopping (similar to the Landweber iteration) and the architecture \cite{dittmer2019}. The drawback is that it is not clear how to choose when to stop. In the original work, they do it using a validation set and select the number of iterations that performs the best on average in terms of PSNR. 

The prior is related to the implicit structural bias of this kind of deep convolutional networks. In the original DIP paper \cite{ulyanov2018dip} and more recently in \cite{chakrabarty2019spectral, heckel2019denoising}, they show that convolutional image generators, optimized with gradient descent, fit \textit{natural} images faster than noise and learn to construct them from low to high frequencies. This effect is illustrated in Figure~\ref{fig:ellipses_dip_iterations}.

\begin{figure}[h]
    \centering
    \input{figures/dip.tex}
    \caption{The figure illustrates the DIP approach. A randomly initialized U-Net-like network is fed with fixed Gaussian noise. The weights are optimized by a gradient descent method to minimize the data discrepancy of the output of the network. We use $128$ channels on every layer, and some have the concatenated skip channels additionally. In our case, we always use $4$ or $0$ skip channels.}
    \label{fig:dip-method}
\end{figure}
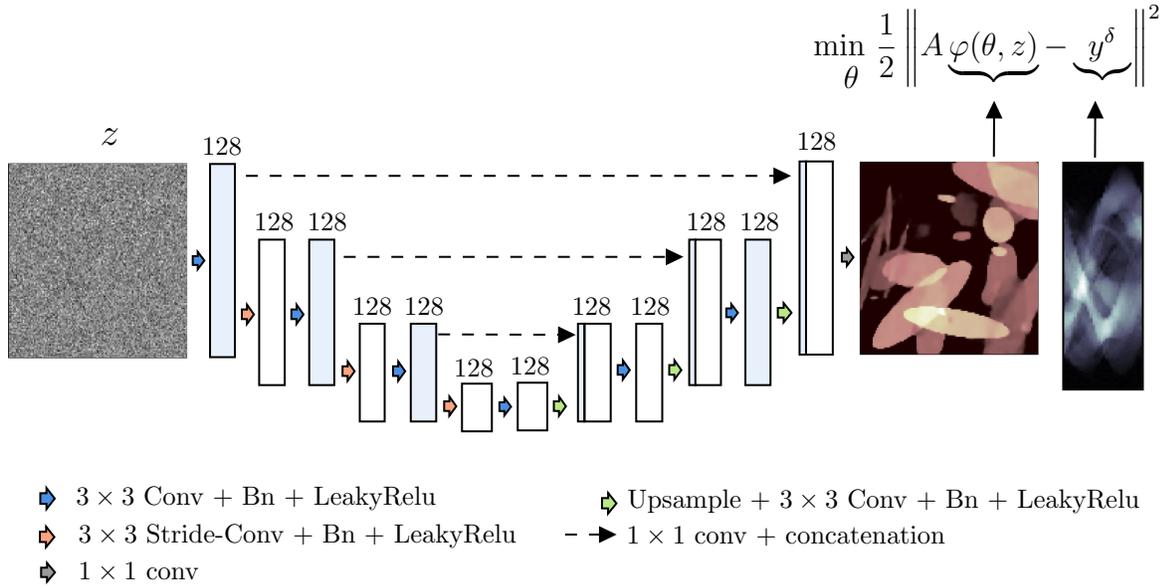

\begin{figure}[h]
\label{fig:dip}
\centering
\includegraphics[scale=0.62]{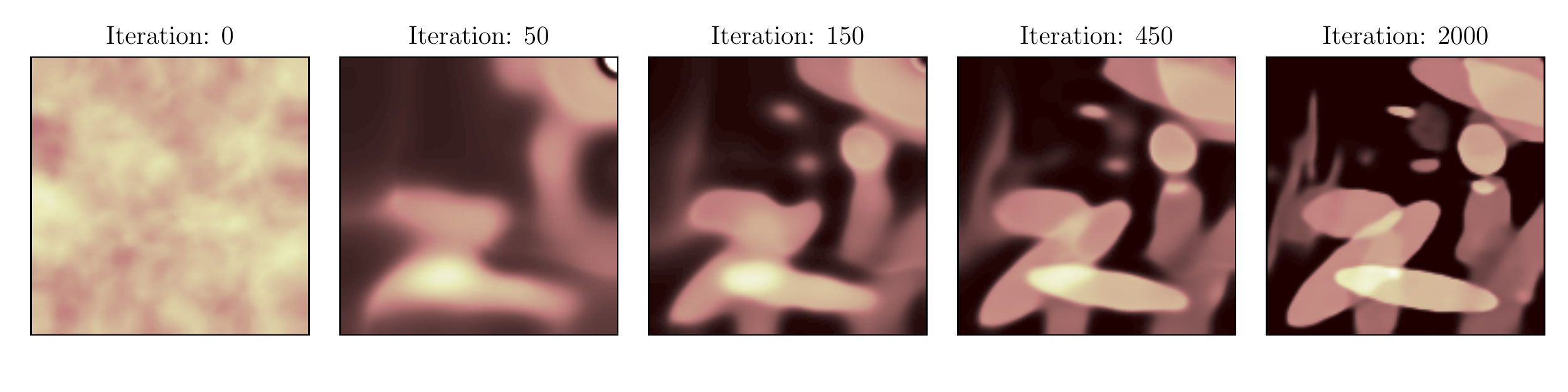}
\caption{Intermediate reconstructions of the DIP approach for CT (Ellipses dataset). At the beginning the coefficients are randomly initialized from a prior distribution. The method starts reconstructing the image from global to local details.}
\label{fig:ellipses_dip_iterations}
\end{figure}

\subsubsection{Related work}
The Deep Image Prior approach has inspired many other researchers to improve it by combining it with other methods \cite{Liu2019, mataev2019deepred, veen2018compressed},  to use it for a wide range of applications \cite{g2018doubledip, gong2018pet, hoyer2019neural, jin2019timedependent} and to offer different perspectives and explanations of why it works \cite{chakrabarty2019spectral, Cheng2019, dittmer2019}. In \cite{mataev2019deepred}, they bring in the concept of Regularization by Denoising (RED). They show how the two (DIP and RED) can be merged into a highly effective unsupervised recovery process. Another series of works, also add explicit priors but on the weights of the network. In \cite{veen2018compressed}, they do it in the form of a multi-variate Gaussian but learn the covariance matrix and the mean using a small dataset. In \cite{Cheng2019}, they introduce a Bayesian perspective on the DIP by also incorporating a prior on the weights $\theta$ and conduct the posterior inference using stochastic gradient Langevin dynamics (SGLD).

So far, the DIP has been used for denoising, inpainting, super-resolution, image decomposition \cite{g2018doubledip}, compressed sensing \cite{veen2018compressed}, PET \cite{gong2018pet}, MRI \cite{jin2019timedependent} among other applications. A similar idea \cite{hoyer2019neural} was also used for structural optimization, which is a popular method for designing objects such as bridge trusses, airplane wings, and optical devices. Rather than directly optimizing densities on a grid, they instead optimize the parameters of a neural network which outputs those densities.

\subsubsection{Network architecture}

In the paper by Ulyanov \etal \cite{ulyanov2018dip}, several architectures were considered, for example, ResNet, Encoder-Decoder (Autoencoder) and a U-Net. For inpainting big holes, the Autoencoder with $\text{depth}=6$ performed best, whereas for denoising a modified U-Net achieved the best results. The regularization happens mainly due to the architecture of the network, which reduces the search space but also influences the optimization process to find \textit{natural} images. Therefore, for each application, it is crucial to choose the appropriate architecture and to tune hyper-parameters, such as the network's depth and the number of channels per layer. Optimizing the hyper-parameters is the most time-consuming part. In Figure~\ref{fig:ellipses_architectures} we show some reconstructions from the Ellipses dataset with different hyper-parameter choices. In this case, it seems that the U-Net without skip connections and depth $5$ (Encoder-Decoder) achieves the best performance. One can see that when the number of channels is too low, the network does not have enough representation power. Also, if there are no skip channels, the higher the number of scales (equivalent to the depth), the more the regularization effect. The extraordinary success of this approach demonstrates that the architecture of the network has a significant influence on the performance of deep learning approaches that use similar kinds of networks.

\begin{figure}[h]
\centering
\includegraphics[scale=0.67]{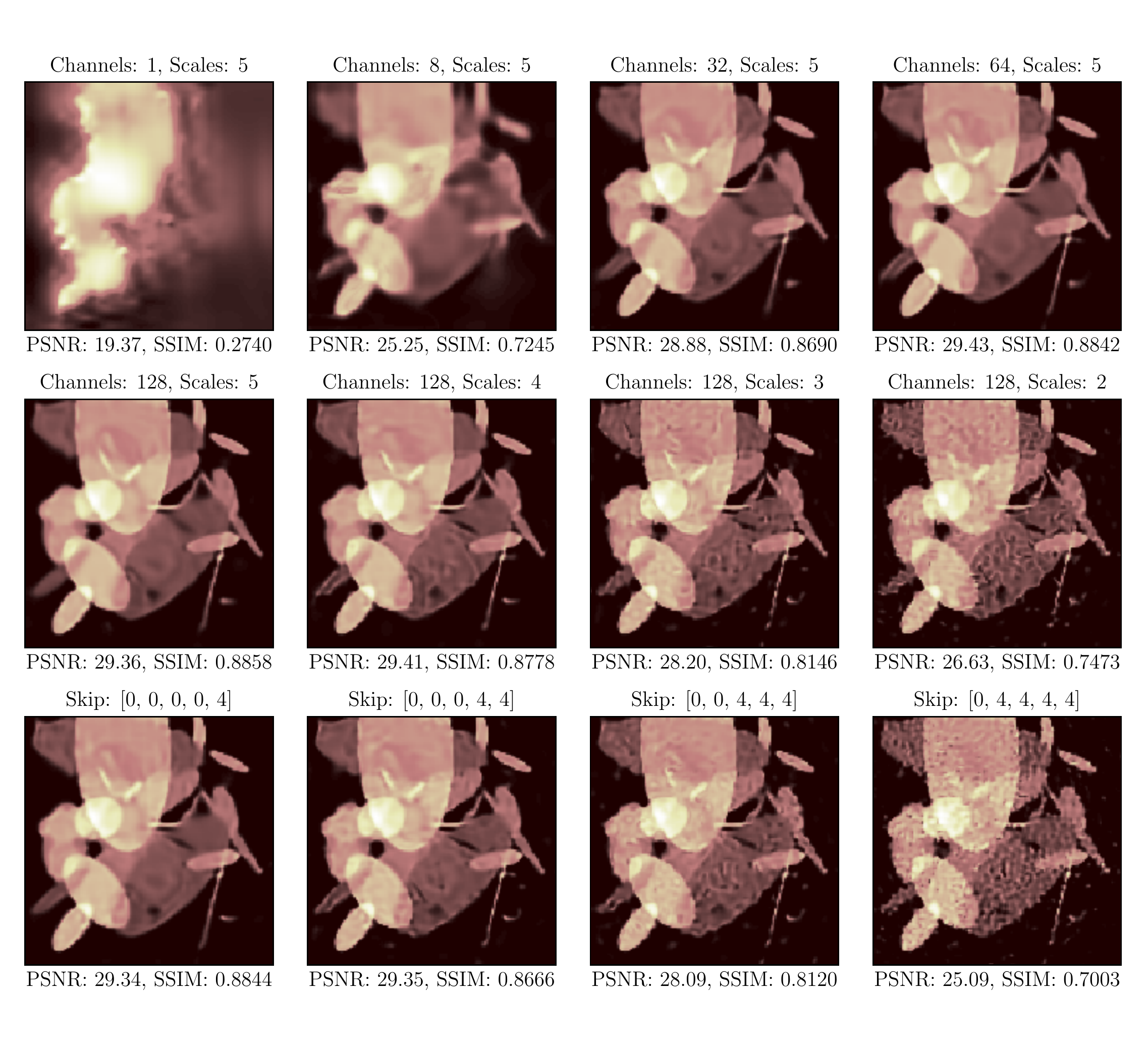}
\caption{CT reconstructions after $5000$ iterations using the DIP with a U-Net architecture and different scales (depths), channels per layer (the network has the same number of channels at every layer) and number of skip connections (the first two rows do not use skip connections, i.e., skip: $[0,\, 0,\, 0,\, 0,\, 0]$). In the last row all reconstructions use $5$ scales and $128$ channels.}
\label{fig:ellipses_architectures}
\end{figure}

\subsubsection{Early-stopping}
As mentioned before, in \cite{ulyanov2018dip}, they show that early stopping has a positive impact on the reconstruction results. They observed (cf. Figure 2) that in some applications, like denoising, the loss decreases fast towards \textit{natural} images, but takes much more time to go towards noisy images. This empirical observation helps to determine when to stop. In Figure \ref{fig:ellipses_dip_losses}, one can observe how the exact error (measured by the PSNR and the SSIM metrics) reaches a maximum and then deteriorates during the optimization process.

\begin{figure}[h]
\centering
\includegraphics[scale=0.685]{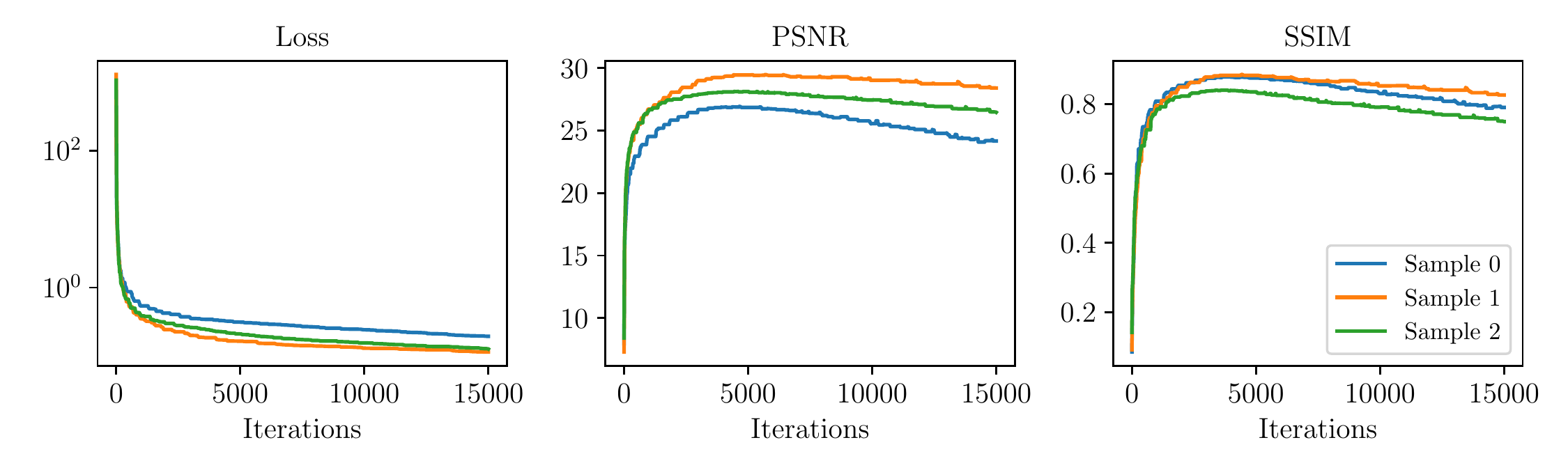}
\caption{Training loss and true error (PSNR and SSIM) of CT reconstructions using the DIP approach. The training was done over 15000 iterations and the architecture is an Encoder-Decoder with $5$ scales and $128$ channels per layer.}
\label{fig:ellipses_dip_losses}
\end{figure}

\section{Deep Image Prior and classical regularization}
\label{sec:regularization}
In this section we analyze the DIP in combination with classical regularization, i.e., we include a regularization term $J:X\to \R \cup \lbrace \infty\rbrace$, such as TV. We give necessary assumptions under which we are able to obtain standard guarantees in inverse problems, such as existence of a solution, convergence, and convergence rates.

In the general case, we consider $X$ and $Y$ to be Banach spaces, and $A:X \to Y$ a continuous linear operator. To simplify notation, we use $\varphi(\cdot)$ instead of $\varphi(\cdot, z)$, since the input to the network is fixed. Additionally, we assume that $\Theta$ is a Banach space, and $\varphi:\Theta \to X$ is a continuous mapping.

The proposed method aims at finding
\begin{equation}
\label{eq:dip_classical}
\theta_\alpha^\delta \in \argmin_{\theta \in \Theta} \disc (A\varphi(\theta), y^\delta) + \alpha J(\varphi(\theta))\ ,
\end{equation}
to obtain
\begin{equation}
\T_\alpha(y^\delta) := \varphi(\theta^\delta_\alpha)\ ,
\end{equation}
for $\alpha > 0$.

With this approach, we get rid of the need for early stopping, i.e., the need to find an optimal number of iterations. Still, we introduce the problem of finding an optimal $\alpha$, which is a classical issue in inverse problems. These problems are similar since both choices depend on the noise level of the observation data. The higher the noise is, the higher the value of $\alpha$ or the smaller the number of iterations for obtaining optimal results.

If the range of $\varphi$ is $\Omega := \mathrm{rg}(\varphi) = X$, i.e., 
\begin{equation}
\forall\ x \in X:\ \exists\ \theta \in \Theta\ s.t\ \varphi(\theta) = x,
\end{equation}
this is equivalent to the standard variational approach in \Eref{eq:variational}. However, although the network can fit some noise, it cannot fit, in general, any arbitrary $x \in X$. This depends on the chosen architecture, and it is mainly because we do not use any fully connected layers. Nevertheless, the minimization in (\ref{eq:dip_classical}) is similar to the setting in \Eref{eq:variational}, if we restrict the domain of $A$ to be $\widetilde{\D}(A) := \D(A) \cap \Omega$. I.e.,

\begin{equation}
\T_\alpha(y^\delta) \in \arg\min_{x\in \widetilde{\D}} \disc(Ax,\, y^\delta) + \alpha J(x),
\end{equation}
where $\widetilde{\D} := \widetilde{\D}(A)\cap \D(J)$. If the following assumptions are satisfied, then all the classical theorems, namely well-posedness, stability, convergence, and convergence rates, still hold, cf.  \cite{Hofmann2007}.

\begin{assumption}
\label{as:assumption1}
The range of $\varphi$, namely $\Omega$, is closed, i.e., if there is a convergent sequence $\{x_k\} \subset \Omega$ with limit $\tilde{x}$, it holds $\tilde{x} \in \Omega$.
\end{assumption}

\begin{definition} An element $x^\dagger \in \widetilde{\D}$ is called a $J$-minimizing solution if $Ax^\dagger=y^\dagger$ and $\forall\, x \in \widetilde{\D}:\, J(x^\dagger) \leq J(x)$, where $y^\dagger$ is the perfect noiseless data.
\end{definition}

\begin{assumption}
\label{as:assumption2}
There exists a $J$-minimizing solution $x^\dagger \in \widetilde{\D}$ and $J(x^\dagger) < \infty$.
\end{assumption}

Assumption~\ref{as:assumption1} guarantees that the restricted domain of $A$ is closed, whereas Assumption~\ref{as:assumption2} guarantees that there is a $J$-minimizing solution in the restricted domain.

\bigskip

The mapping $\varphi: \Theta \to X$, has a neural network structure, with a fixed input $z \in \R^{n_0}$, and can be expressed as a composition of affine mappings and activation functions
\begin{equation}
\varphi = \sigma_L \circ \K_L \circ \cdots \circ \sigma_2 \circ \K_2 \circ \sigma_1 \circ \K_1\ ,
\end{equation}
where $\K_i(x) := \Gamma_i x + b_i$, $\Gamma_i \in G_i \subseteq \R^{n_i \times n_{i-1} }$, $b_i \in B_i \subseteq \R^{n_i}$, $\theta = (\Gamma_L,\, b_L,\, \cdots,\, \Gamma_1,\, b_1) \in G_L \times B_L \cdots \times G_1 \times B_1 = \Theta$ and $\sigma_i:\R^{n_i}\to\R^{n_i}$. In the following we analyze under which conditions we can guarantee that the range of $\varphi$ (with respect to $\Theta$) is closed.

\begin{definition}
An activation function $\sigma: \R^n \to \R^n$ is valid, if it is continuous, monotone, and bounded, i.e., there exist $c>0$ such that $\forall x\in X: \nsqnorm{\sigma(x)} \leq c\nsqnorm{x}$.
\end{definition}

\begin{lemma}
\label{le:lemma1}
Let $\varphi$ be a neural network $\varphi:\Theta \to X$ with $L$ layers. If $\Theta$ is a compact set, and the activation functions $\sigma_i$ are \textit{valid}, then the range of $\varphi$ is closed.
\end{lemma}

\begin{proof}
In order to prove the result we show that the range after each layer of the network is compact. \\
i) Let the set $V_i = \{\Gamma u:\ \Gamma \in G_i,\, u \in U_i \subset \R^{n_{i-1}}\}$, where $U_i$ is bounded and closed. From the compactness of $\Theta$ it follows that $G_i$ is also bounded and closed, therefore, $V_i$ is also bounded. Let the sequence $\{\Gamma^{(k)}u^{(k)}\}$, with $\Gamma^{(k)} \in G_i$ and $u^{(k)} \in U_i$, converge to $v$. Since $\{\Gamma^{(k)}\}$ and $\{u^{(k)}\}$ are bounded, there is a subsequence $\{\overline{\Gamma}^{(k)}\bar{u}^{(k)}\}$, where both $\{\overline{\Gamma}^{(k)}\}$ and $\{\bar{u}^{(k)}\}$ converge to $\overline{\Gamma} \in G_i$ and $\bar{u}\in U_i$ respectively. It follows that $\{\overline{\Gamma}^{(k)}\bar{u}^{(k)}\}$ converges to $\overline{\Gamma}\bar{u}$, therefore, $v= \overline{\Gamma}\bar{u} \in V_i$, which shows that $V_i$ is closed.\\
ii) From i) and the fact that $B_i$ is also compact it follows that the set $V_i = \{\Gamma u + b:\ \Gamma \in G_i \subset \R^{n_i \times n_{i-1}},\, u \in U_i \subset \R^{n_{i-1}}, b \in B_i \subset \R^{n_i}\}$ is still closed and bounded.\\
iii) It is easy to show that if the pre-image of a \textit{valid} activation $\sigma$ is compact, then its image is also compact.\\
In the first layer, $V_0 = \{z\}$; thus, it can be shown by induction that the range of $\varphi:\Theta \to X$ is closed.
\end{proof}

All activation functions commonly used in the literature, for example, sigmoid, hyperbolic tangent, and piece-wise linear activations, are \textit{valid}. The bounds on the weights of the network can be ensured by clipping the weights after each gradient update. In our implementation of the DIP approach, we use a sufficiently large bound and empirically check that Assumption~\ref{as:assumption2} holds.

\begin{remark}
An alternative condition to the bound on the weights is to use only \textit{valid} activation functions with closed range, for example, ReLU or leaky ReLU. However, it wouldn't be possible to use sigmoid or hyperbolic tangent. In our experiments we observed that having a sigmoid activation in the last layer performs better than having a ReLU.
\end{remark}

\section{Deep Image Prior with initial reconstruction}
\label{sec:dip_initial}

In this section, we propose a new method based on the DIP approach. It takes the result from any end-to-end learned method $\T: Y \to X$ as initial reconstruction and further enforces data consistency by optimizing over its deep-neural parameterization.

\begin{definition}[Deep-neural parameterization]
Given an untrained network $\varphi: \Theta \times Z \to X$ and a fixed input $z \in Z$, the deep-neural parameterization of an element $x \in X$ with respect to $\varphi$ and $z$ is 
\begin{equation}
\theta_x \in \argmin_{\theta \in \Theta} \sqnorm{\varphi(\theta,\, z) - x}\ .
\end{equation}
\end{definition}

The projection onto the range of the network is possible because of the result of Lema~\ref{le:lemma1}, i.e., the range is closed. If $\varphi$ is a deep convolutional network, for example, a U-Net, the deep-neural parameterization has similarities with other signal representations, such as the Wavelets and Fourier transforms \cite{hoyer2019neural}. For image processing, such domains are usually more convenient than the classical pixel representation.

As shown in Figure~\ref{fig:space}, one way to enforce data consistency is to project the initial reconstruction into the set where $\nsqnorm{Ax - y^\delta} \leq \delta$. The puzzle is that due to the ill-posedness of the problem, the new solution (red point) will very likely have artifacts. The proposed approach first obtains the deep-neural parameterization $\theta_0$ of the initial reconstruction $\T(y^\delta)$ and then use it as starting point to minimize
\begin{equation}
\label{eq:diptv_loss}
\loss(\theta) := \sqnorm{A\varphi(\theta, z) - y^\delta} + \alpha J(\varphi(\theta, z)),
\end{equation}
over $\theta$ via gradient descent. The iterative process is conveyed until $\nsqnorm{A\varphi(\theta, z) - y^\delta} \leq \delta$ or for a given fixed number of iterations $K$ determined by means of a validation dataset. This approach seems to force the reconstruction to stay close to the set of \textit{natural} images because of the structural bias of the deep-neural parameterization. The procedure is listed in Algorithm~\ref{alg:new_method} and a graphical representation is shown in Figure~\ref{fig:space}. 

The new method $\hat{\T}: Y \to X$ is similar to other image enhancement approaches. For example, related methods \cite{Donoho1994}, first compute the wavelet transform (parameterization), and then repeatedly do smoothing or shrinking of the coefficients (further optimization).

\begin{figure}
    \centering
    \includegraphics{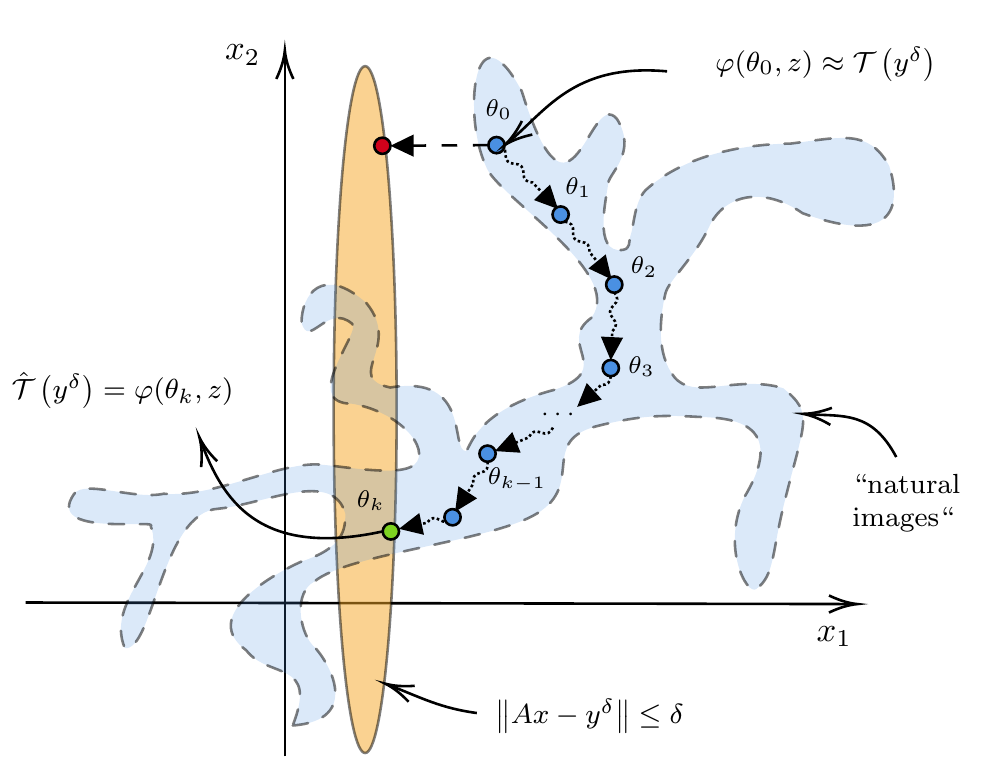}
    \caption{Graphical illustration of the DIP approach with initial reconstruction. The blue area refers to an approximation of some part of the space of \textit{natural} images.}
    \label{fig:space}
\end{figure}

\begin{algorithm}[t]
\caption{Deep Image Prior with initial reconstruction}
\begin{algorithmic}[1]
\setlength{\baselineskip}{1.5\baselineskip}
\State {${\displaystyle x_0 \gets \T(y^\delta)}$}
\State {${\displaystyle z \gets \mathrm{noise}}$}
\State {${\displaystyle \theta_0 \in \argmin_\theta \sqnorm{\varphi(\theta, z) - x_0}}$}
\For{$k \gets 0$ to $K-1$}
    \State {${\displaystyle \omega \in \partial \loss(\theta_k)}$}
    \State {${\displaystyle \theta_{k+1} \gets \theta_k - \eta \omega}$}
\EndFor
\State {${\displaystyle\hat{\T}(y^\delta)} \gets \varphi(\theta_k, z)$}
 
\end{algorithmic}
\label{alg:new_method}
\end{algorithm}

\section{Benchmark setup and results}
\label{sec:benchmark}

For the benchmark, we implemented the end-to-end learned methods described in Section~\ref{sec:learned_methods}. We trained them on different data-sizes and compared them with classical methods, such as FBP and TV regularization, and with the proposed methods. The datasets we use were recently released to benchmark deep learning methods for CT reconstruction \cite{leuschner2019lodopabct}. They are accessible through the DIV$\alpha\ell$ python library \cite{leuschner2019dival}. We also provide the code and the trained methods in the following GitHub repository: \url{https://github.com/oterobaguer/dip-ct-benchmark}.

\subsection{The LoDoPaB-CT Dataset}

The low-dose parallel beam (LoDoPaB) CT dataset \cite{leuschner2019lodopabct} consists of more than \num{40000} two-dimensional CT images and corresponding simulated low-intensity measurements.
Human chest CT reconstructions from the LIDC/IDRI database \cite{armato2011lidc_idri} 
are used as virtual ground truth.
Each image has a resolution of $\num{362}\times\num{362}$ pixels.
For the simulation setup, a simple parallel beam geometry with \num{1000} angles and \num{513} projection beams is used.
To simulate low intensity, Poisson noise corresponding to a mean photon count of \num{4096} photons per detector pixel before attenuation is applied to the projection data.
We use the standard dataset split defining in total \num{35820} training pairs, \num{3522} validation pairs and \num{3553} test pairs.

\subsection{Ellipses Dataset}

As a synthetic dataset for imaging problems, random phantoms of combined ellipses are commonly used.
We use the \texttt{'ellipses'} standard dataset from the DIV$\alpha\ell$ python library (as provided in version 0.4) \cite{leuschner2019dival}.
The images have a resolution of $\num{128}\times\num{128}$ pixels.
Measurements are simulated with a parallel beam geometry with only \num{30} angles and \num{183} projection beams.
In addition to the sparse-angle setup, moderate Gaussian noise with a standard deviation of \SI{2.5}{\%} of the mean absolute value of the projection data is added to the projection data.
In total, the training set contains \num{32000} pairs, while the validation and test set consist of \num{3200} pairs each.

\subsection{Implementation details}

For the DIP with initial reconstruction, we used the learned primal-dual, which we consider to be state of the art for this task (see the results in Figure~\ref{fig:ellipses_performace}). For each data-size, we chose different hyper-parameters, namely the step-size $\eta$, the TV regularization parameter $\gamma$, and the number of iterations $K$, based on the available validation dataset ($3$ data-pairs for the smallest size).

Minimizing $\loss(\theta)$ in \eref{eq:diptv_loss} is not trivial because TV is not differentiable. In our implementation we use the PyTorch automatic differentiation framework \cite{paszke2017automatic} and the ADAM \cite{kingma2014adam} optimizer. For the Ellipses dataset we use the $\ell_2$-discrepancy term, whereas for the LoDoPaB we use the Poisson loss.

\subsection{Numerical results}
\label{sec:benchmark_results}

We trained all the methods with different dataset sizes. For example, $\SI{0.1}{\%}$ on the Ellipses dataset means we trained the model with $\SI{0.1}{\%}$ ($32$ data-pairs) of the available training data and $\SI{0.1}{\%}$ ($3$ data-pairs) of the validation data. Afterward, we tested the performance of the method on $100$ samples of the test dataset. More details are depicted in \ref{ap:details}.

As expected, on both datasets, the fully learned method (iRadonMap) requires much data to achieve acceptable performance. On the Ellipses dataset, it outperformed TV using $\SI{100}{\%}$ of the data, whereas on the LoDoPaB dataset, it performed just slightly better than the FBP. The learned post-processing (FBP+UNet) required much less data. It outperformed TV with only $\SI{10}{\%}$ of the Ellipses dataset and $\SI{0.1}{\%}$ of the LoDoPaB dataset. On the other hand, we find that the learned primal-dual is very data efficient and achieved the best performance. On both datasets, it outperformed TV, trained with only $\SI{0.1}{\%}$ (32 data-pairs) and $\SI{0.01}{\%}$ (4 data-pairs from the same patient) of the Ellipses and LoDoPaB datasets respectively. In Figure~\ref{fig:lodopab_data}, we show some results from the test set.

The DIP+TV approach achieved the best results among the data-free methods. On average, it outperforms TV by $1\,\si{dB}$, and $2\,\si{dB}$ on the Ellipses and LoDoPaB datasets respectively. In Figure~\ref{fig:ellipses_dip_vs_diptv}, it can be observed that TV tends to produce flat regions but also produces high staircase effects on the edges. However, the combination with DIP seems to produce more realistic edges. For the first two smaller data-sizes, it performs better than all the end-to-end learned methods.

\begin{figure}[h]
\centering
\includegraphics[scale=0.68]{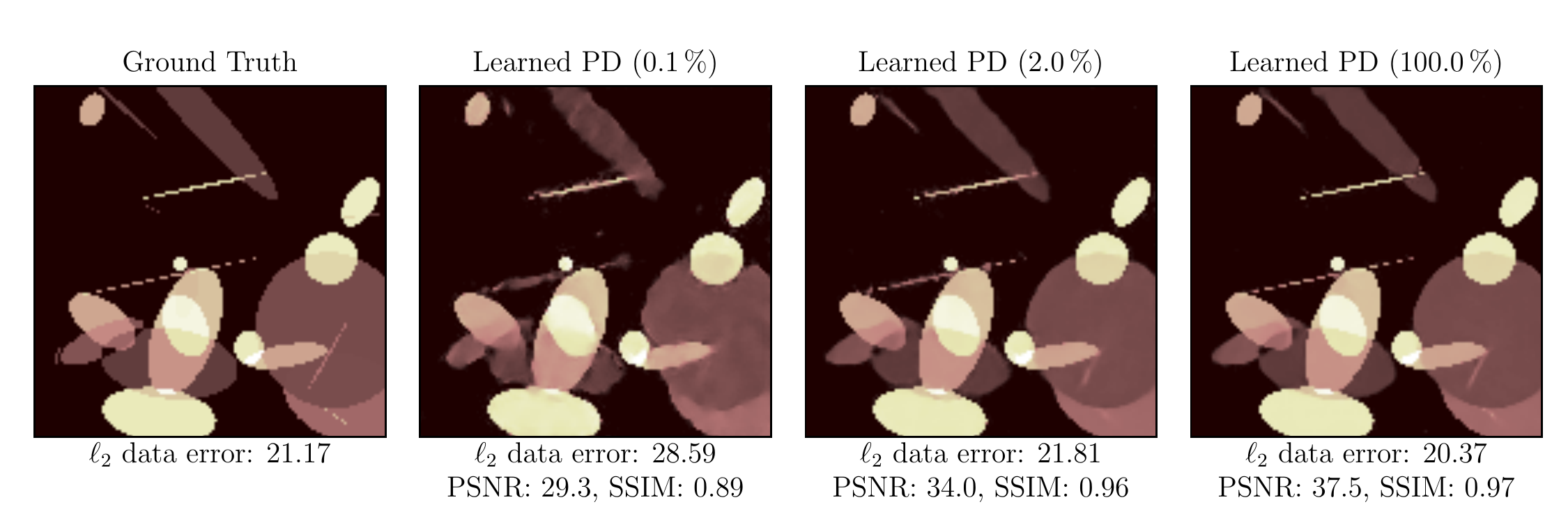}
\includegraphics[scale=0.68]{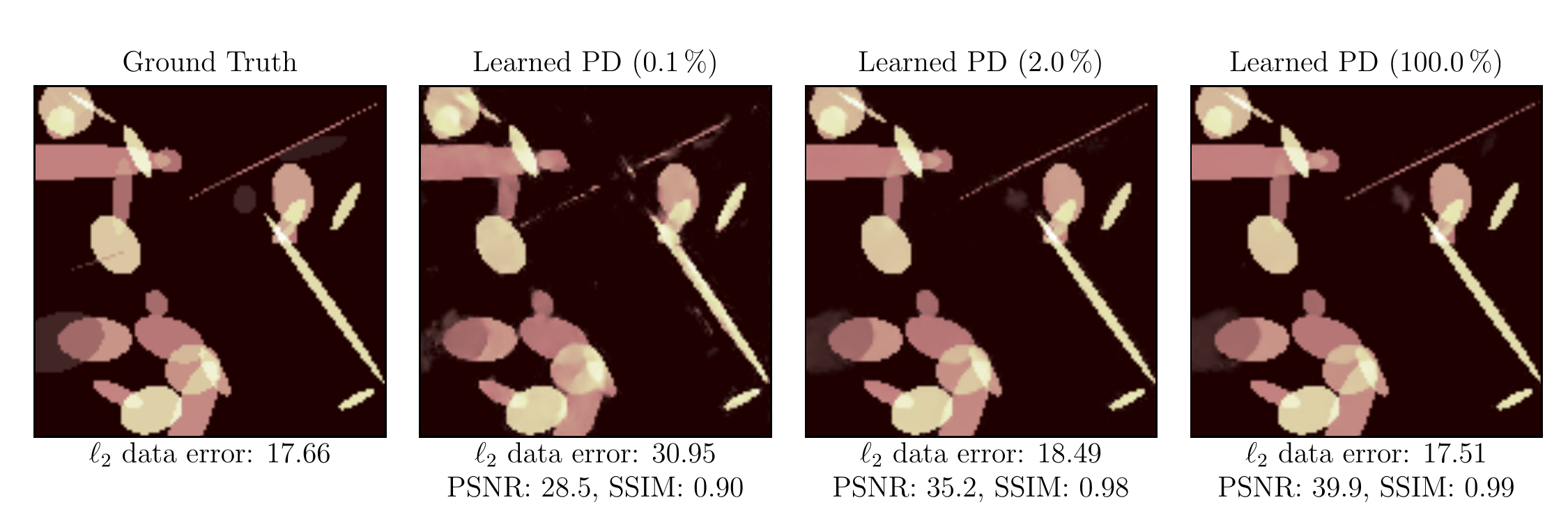}
\includegraphics[scale=0.9]{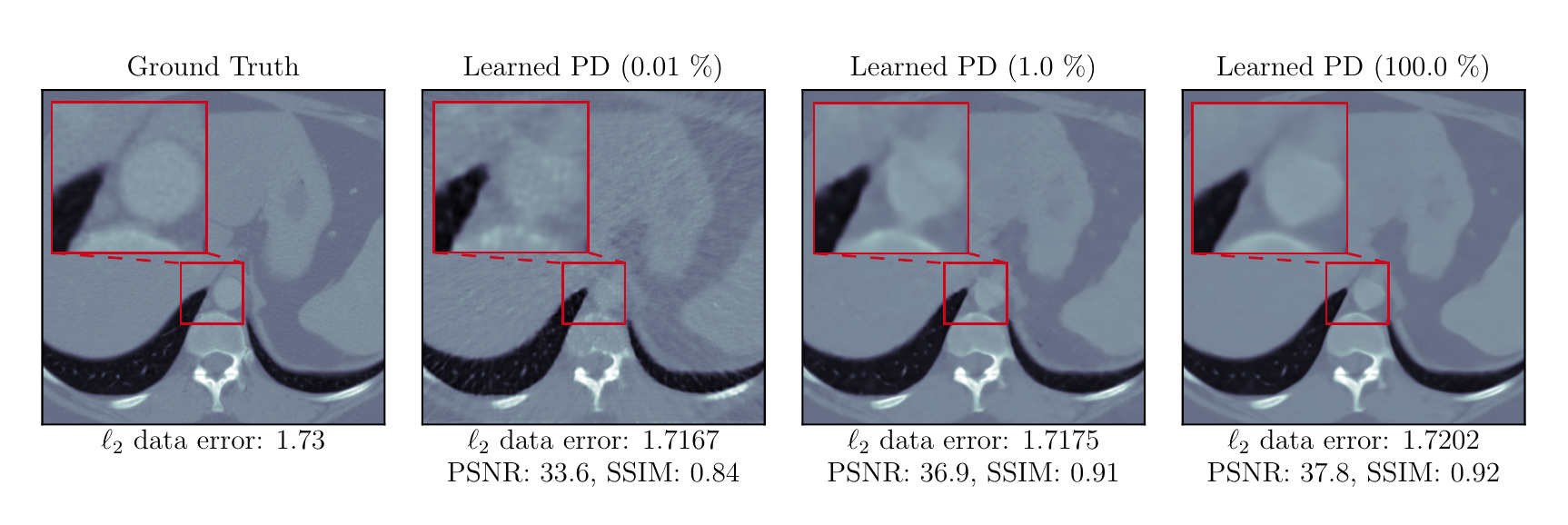}
\includegraphics[scale=0.9]{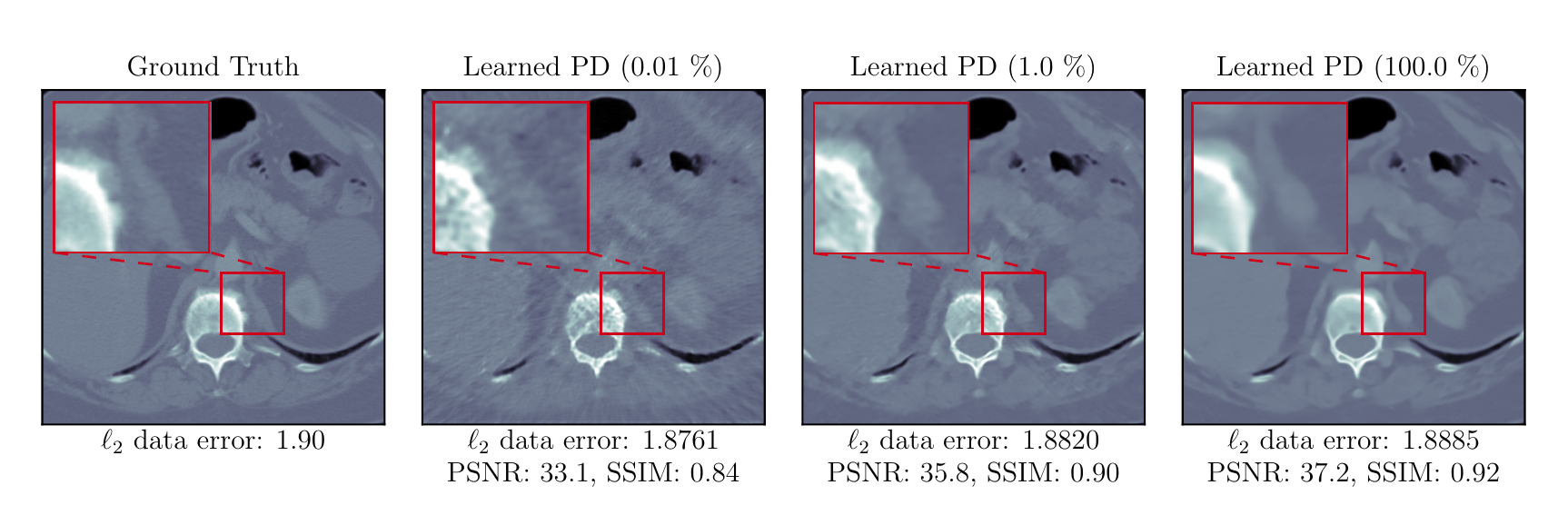}
\caption{Reconstructions using the learned primal-dual method trained with different amounts of data.}
\label{fig:lodopab_data}
\end{figure}

\begin{figure}
\centering
\includegraphics[scale=0.685]{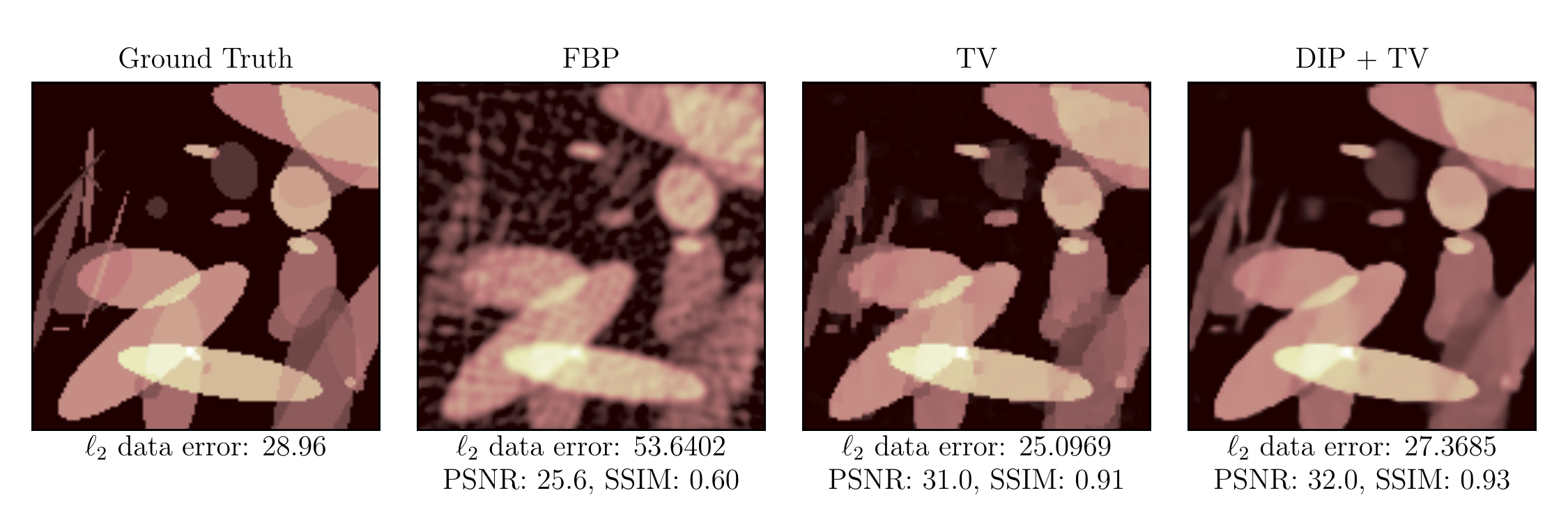}
\includegraphics[scale=0.685]{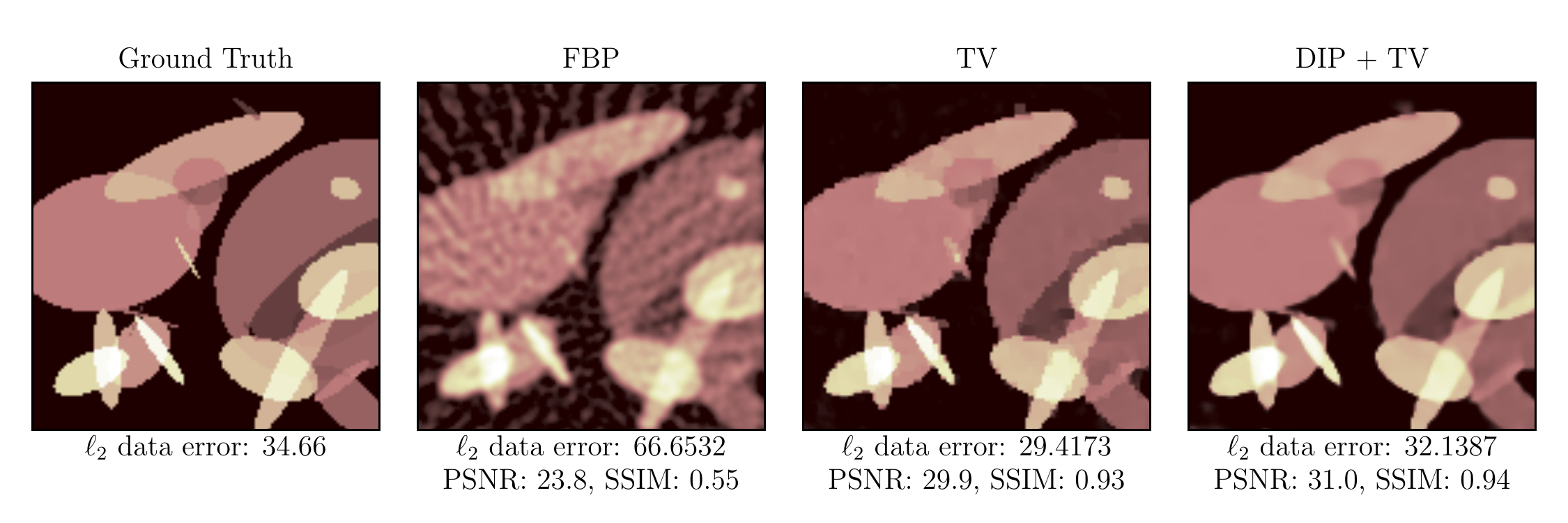}
\includegraphics[scale=0.9]{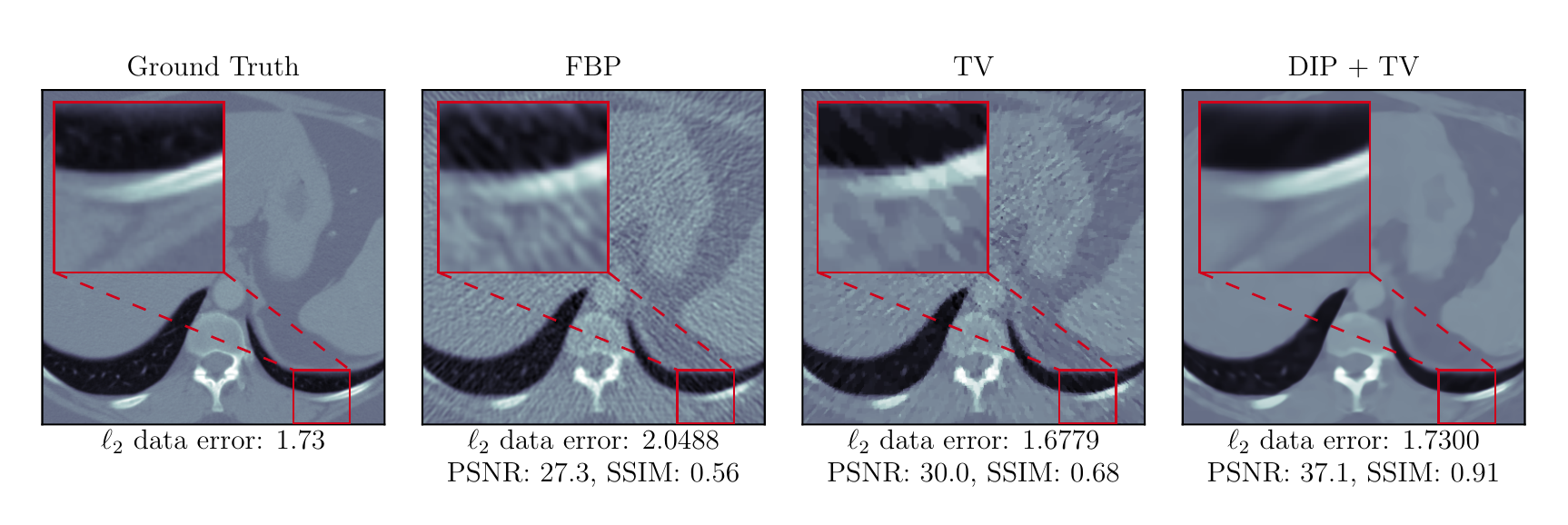}
\includegraphics[scale=0.9]{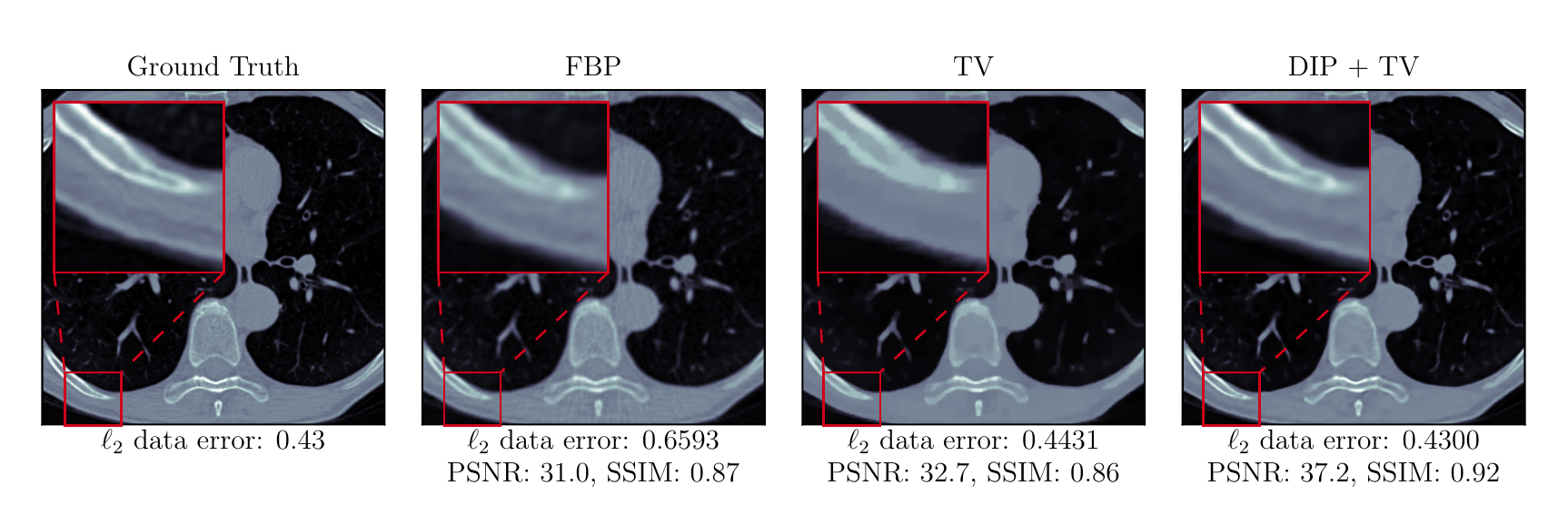}
\caption{Reconstruction obtained with the Filtered Back Projection (FBP) method, isotropic TV regularization and the Deep Image Prior (DIP) approach combined with TV.}
\label{fig:ellipses_dip_vs_diptv}
\end{figure}

The Deep Image Prior in combination with the learned primal-dual achieved the best results on the low-data regime. For the Ellipses dataset, it improved the quality of the reconstructions up to $\SI{+1}{dB}$ on average. However, for dataset sizes bigger than $\SI{2}{\%}$, the method did not yield any significant change. On the LoDoPaB data, we did not find a notable improvement. For the smaller sizes, it did improve, but it was just as good as the DIP+TV approach. We believe that this approach is more useful in the case of having sparse measurements, as in the Ellipses dataset.

In Figure~\ref{fig:ellipses_fastdip_example}, we show some reconstructions obtained using this method for the Ellipses dataset, and compare them with the original initial reconstructions. The reconstructions have a better data consistency w.r.t the observed data ($\ell_2$-discrepancy) and higher quality both visually and in terms of the PSNR and SSIM measures. Moreover, it needed fewer iterations than the DIP+TV, even if we also consider the iterations required to obtain the deep-prior/neural parameterization of the first reconstruction. These initial iterations are much faster because they only use the identity operator instead of the Radon transform.

In our setting, for the Ellipses dataset, the DIP+TV approach needs $8000$ iterations to obtain optimal performance in a validation dataset ($5$ ground truth and observation pairs). On the other hand, by using the initial reconstruction, it needs $4000$ iterations with the identity operator and only $1000$ with the Radon transform operator. With an nVidia
GeForce GTX 1080 Ti graphics card, the original DIP takes approx. $6$ min per reconstruction, whereas the proposed method takes $3$ min ($2\times$ speed factor). The used Encoder-Decoder architecture has approx. $2\cdot 10^6$ parameters in total.

\begin{figure}[!tbp]
  \includegraphics[width=\textwidth]{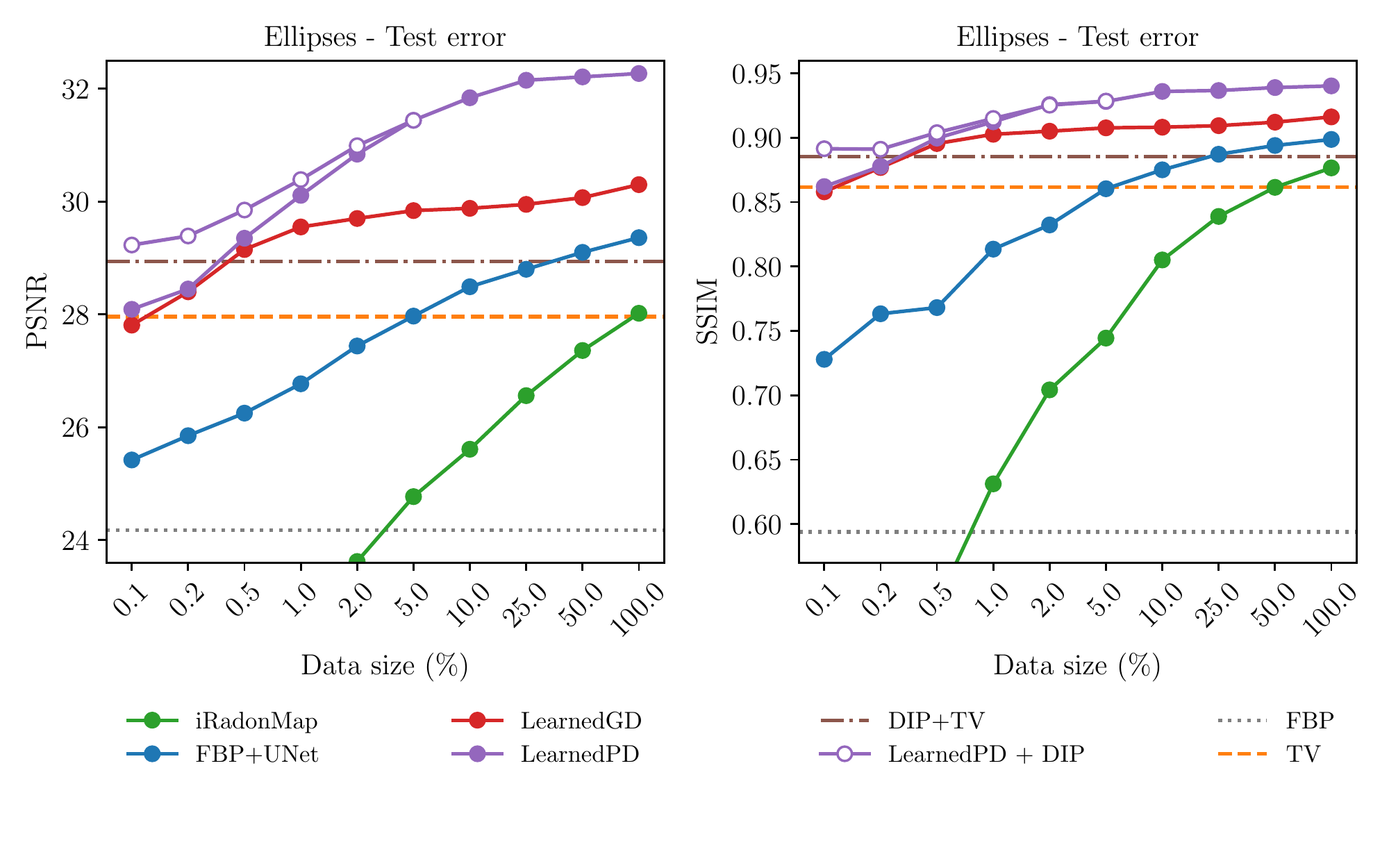}
  \includegraphics[width=\textwidth]{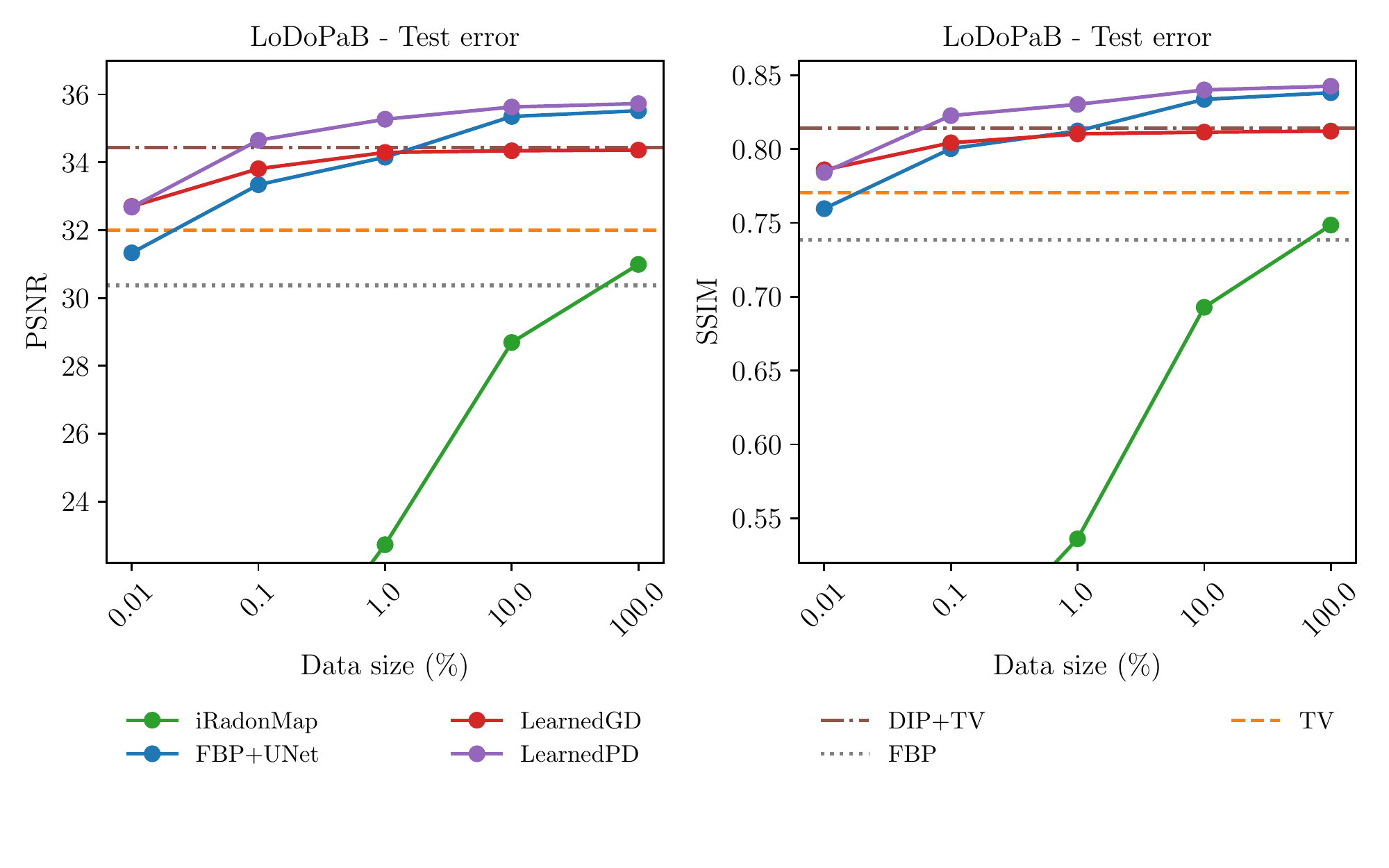}
  \caption{Benchmark results of the compared classical methods (Filtered Back Projection, TV), learned methods (FBP+UNet, iRadonMap, learned gradient descent, learned primal-dual) and the proposed approaches (DIP+TV, learned primal-dual + DIP) on the Ellipses and LoDoPaB standard datasets. The horizontal lines indicate the performance of the data-free methods.}
  \label{fig:ellipses_performace}
\end{figure}

\begin{figure}
\centering
\includegraphics[scale=0.685]{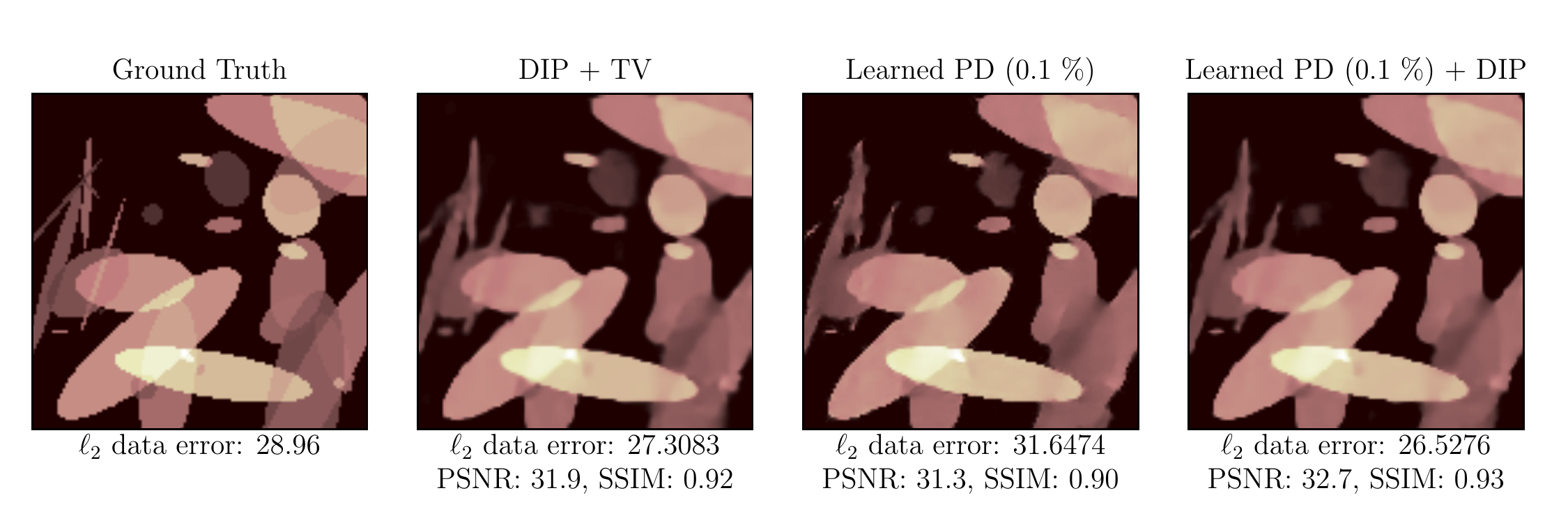}
\includegraphics[scale=0.685]{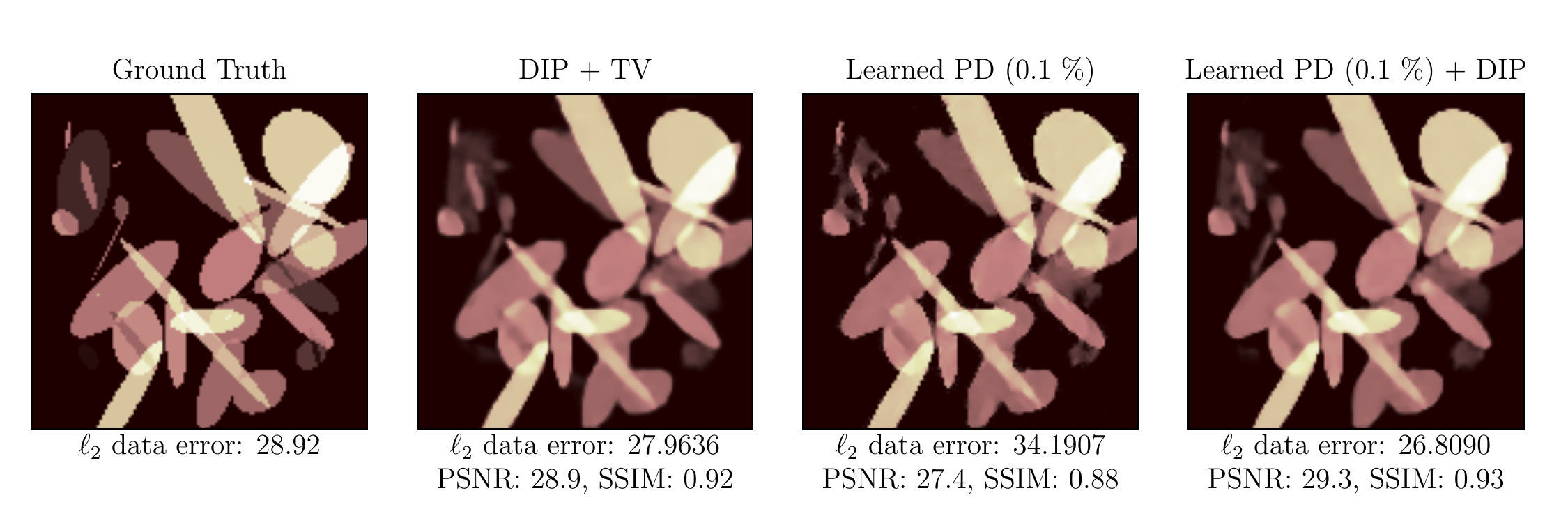}
\includegraphics[scale=0.685]{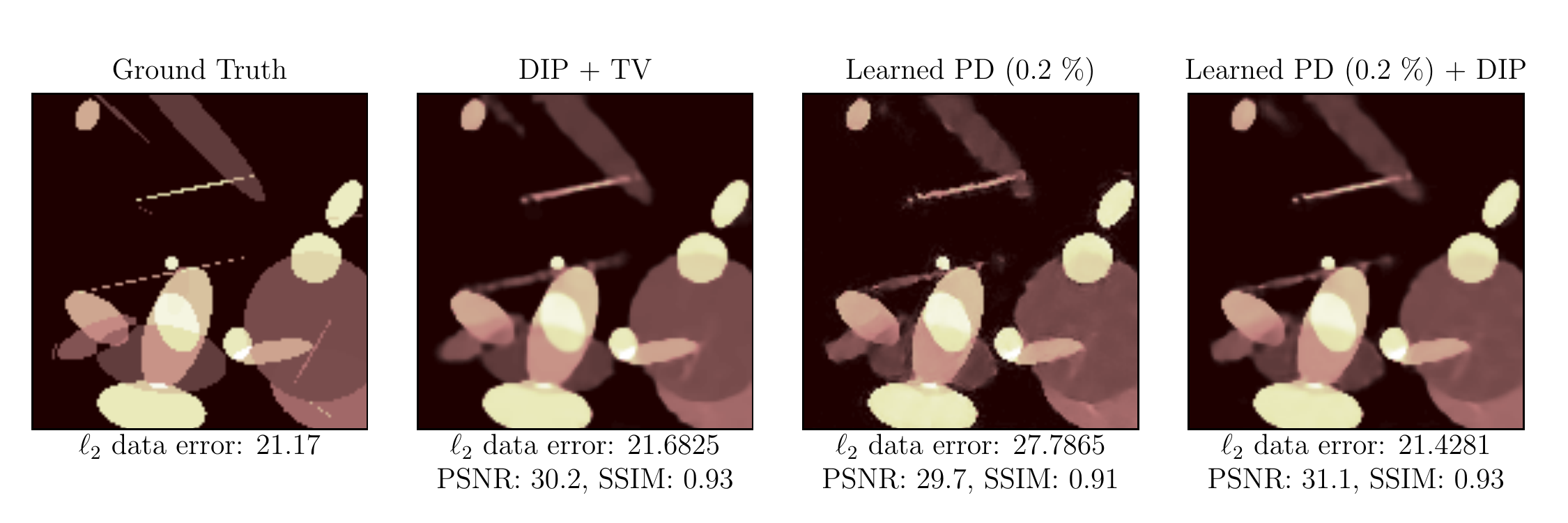}
\includegraphics[scale=0.685]{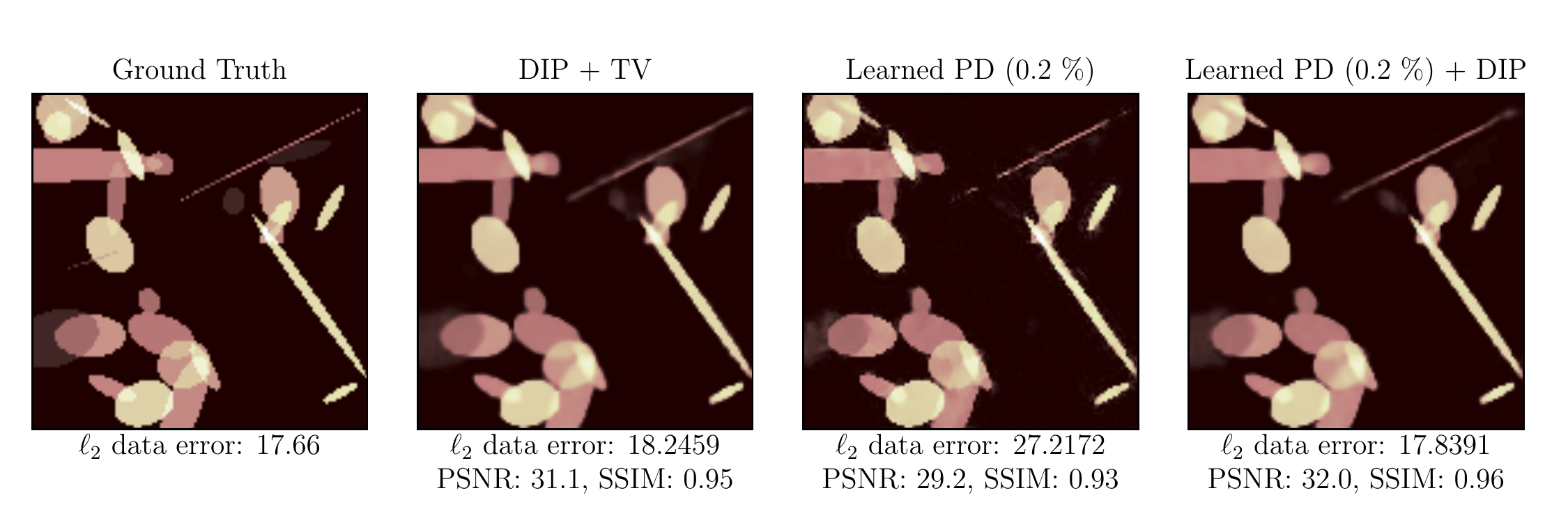}
\caption{Examples of reconstructions obtained with the filtered back projection (FBP), the learned primal-dual method trained with $\SI{0.1}{\%}$ and $\SI{0.2}{\%}$ of the Ellipses dataset ($32$ and $64$ resp. data-pairs) and the DIP approach with initial reconstruction.}
\label{fig:ellipses_fastdip_example}
\end{figure}

\section{Conclusions}

In this work, we study the combination of classical regularization, deep-neural parameterization, and deep learning approaches for CT reconstruction. We benchmark the investigated methods and evaluate how they behave in low-data regimes. Among the data-free approaches, the DIP+TV method achieves the best results. However, it is considerably slow and does not benefit from having a small dataset. On the other hand, the learned primal-dual is very data efficient. Still, it lacks data consistency when not trained with enough data. These issues motivate us to adjust the reconstruction obtained with the learned primal-dual to match the observed data. We solved the puzzle without introducing artifacts through a combination of classical regularization and the DIP. We also derived conditions under which theoretical guarantees hold and showed how to obtain them.

The results presented in this paper offer several baselines for future comparisons with other approaches. Moreover, the proposed methods could be applied to other imaging modalities.

\ack
The authors acknowledge the support by the Deutsche Forschungsgemeinschaft (DFG) within the framework of GRK 2224/1 ``$\pi^3$: Parameter Identification -- Analysis, Algorithms, Applications''. The authors also thank Jonas Adler, Jens Behrmann, S{\"o}ren Dittmer and Peter Maass for useful comments and discussions.

\pagebreak

\section*{References}
\bibliographystyle{plain}
\bibliography{main}

\pagebreak

\appendix

\section{More results}
\begin{figure}[ht]
\centering
\includegraphics[width=0.95\textwidth]{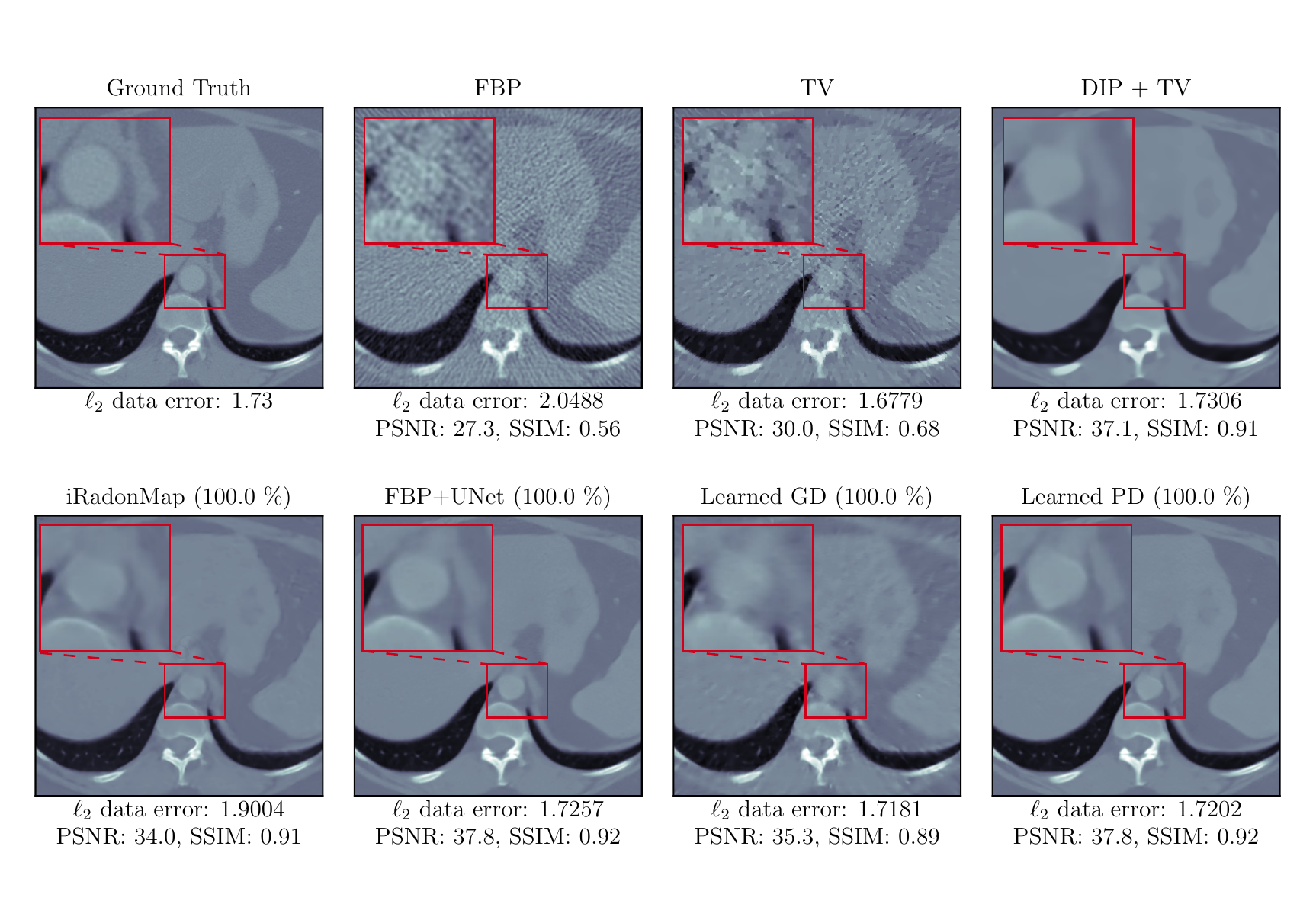}
\includegraphics[width=0.95\textwidth]{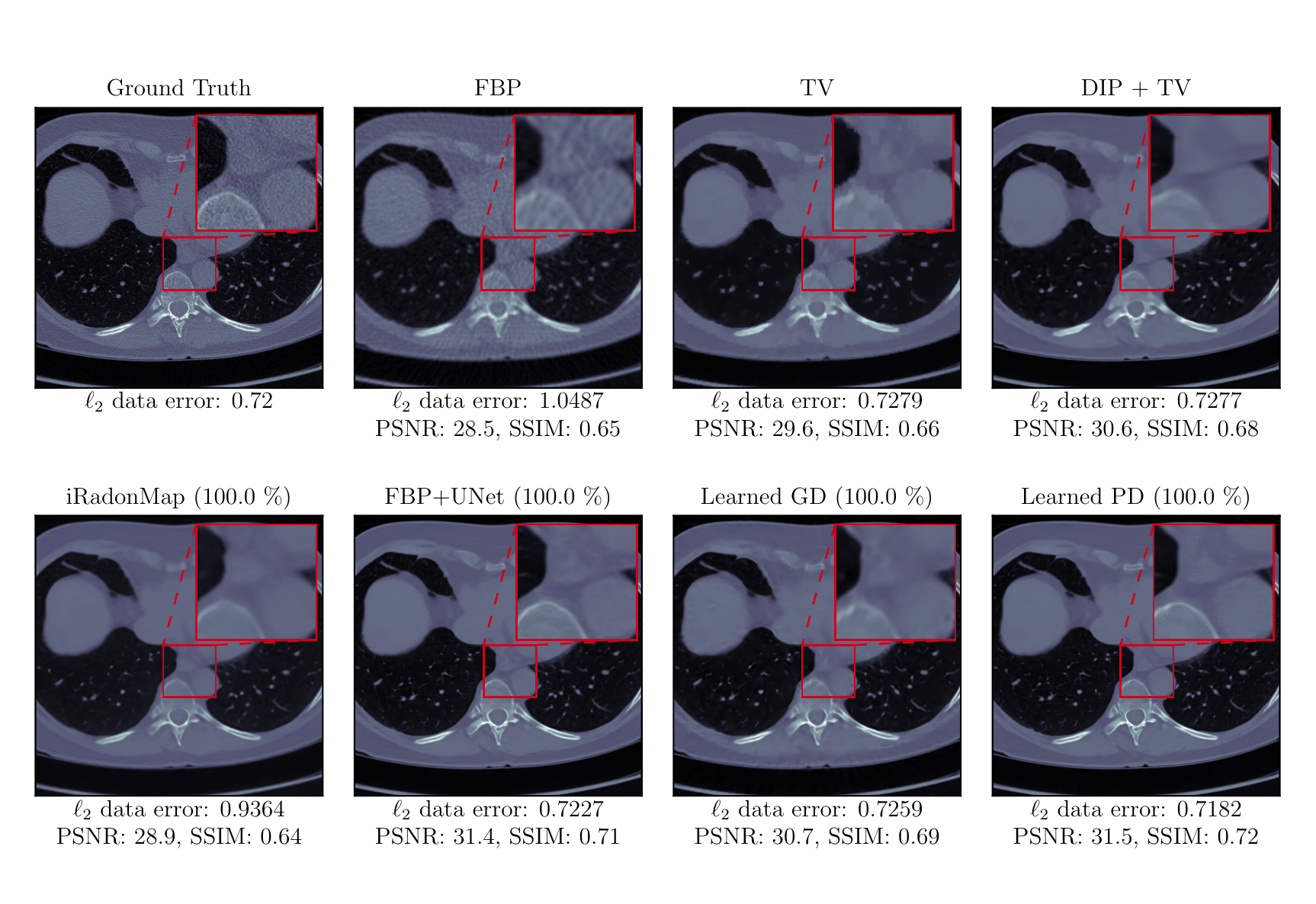}
\caption{Reconstructions using all the analyzed methods for test samples from the LoDoPaB dataset.}
\label{fig:lodopab_all}
\end{figure}

\section{Training details}
\label{ap:details}
\begin{table}[ht]
    \centering
    \input{table_ellipses_data_sizes}
    \caption{The amounts of training and validation pairs from the Ellipses dataset used for the benchmark in Section~\ref{sec:benchmark}.}
\end{table}
\begin{table}[ht]
    \centering
    \input{table_lodopab_data_sizes}
    \caption{The amounts of training and validation pairs from the LoDoPaB dataset used for the benchmark in Section~\ref{sec:benchmark}. The last two lines denote the numbers of patients of whom images are included.}
\end{table}

\end{document}

%% file: figures/parallel_beam.pdf_tex
\begingroup%
  \makeatletter%
  \providecommand\color[2][]{%
    \errmessage{(Inkscape) Color is used for the text in Inkscape, but the package 'color.sty' is not loaded}%
    \renewcommand\color[2][]{}%
  }%
  \providecommand\transparent[1]{%
    \errmessage{(Inkscape) Transparency is used (non-zero) for the text in Inkscape, but the package 'transparent.sty' is not loaded}%
    \renewcommand\transparent[1]{}%
  }%
  \providecommand\rotatebox[2]{#2}%
  \ifx\svgwidth\undefined%
    \setlength{\unitlength}{283.46455078bp}%
    \ifx\svgscale\undefined%
      \relax%
    \else%
      \setlength{\unitlength}{\unitlength * \real{\svgscale}}%
    \fi%
  \else%
    \setlength{\unitlength}{\svgwidth}%
  \fi%
  \global\let\svgwidth\undefined%
  \global\let\svgscale\undefined%
  \makeatother%
  \begin{picture}(1,0.82499999)%
    \put(0,0){\includegraphics[width=\unitlength]{figures/parallel_beam.pdf}}%
    \put(0.51957198,0.42489504){\color[rgb]{0.66666667,0,0}\makebox(0,0)[lb]{\smash{$\varphi$}}}%
    \put(0.53234174,0.51258616){\color[rgb]{0.66666667,0,0}\makebox(0,0)[lb]{\smash{$s$}}}%
    \put(0.7676556,0.14622696){\color[rgb]{0,0,0}\rotatebox{59.99999989}{\makebox(0,0)[lb]{\smash{X-ray source}}}}%
    \put(0.17118289,0.52565148){\color[rgb]{0,0,0}\rotatebox{59.99999989}{\makebox(0,0)[lb]{\smash{detector}}}}%
  \end{picture}%
\endgroup%

%% file: figures/dip.tex
\tikzset{every picture/.style={line width=0.75pt}} 

\begin{tikzpicture}[x=0.75pt,y=0.75pt,yscale=-1,xscale=1]

\draw  [line width=0.75]  (39.5,132.69) -- (128.8,132.69) -- (128.8,230.01) -- (39.5,230.01) -- cycle ;
\draw  [dash pattern={on 4.5pt off 4.5pt}]  (319.92,319.56) -- (342.92,319.56) ;
\draw [shift={(345.92,319.56)}, rotate = 180] [fill={rgb, 255:red, 0; green, 0; blue, 0 }  ][line width=0.08]  [draw opacity=0] (8.93,-4.29) -- (0,0) -- (8.93,4.29) -- cycle    ;
\draw  [fill={rgb, 255:red, 74; green, 144; blue, 226 }  ,fill opacity=0.13 ] (438.14,131.23) -- (441.41,131.23) -- (441.41,228.73) -- (438.14,228.73) -- cycle ;
\draw  [fill={rgb, 255:red, 255; green, 255; blue, 255 }  ,fill opacity=1 ][line width=0.75]  (329.66,213.15) -- (342.9,213.15) -- (342.9,262.5) -- (329.66,262.5) -- cycle ;
\draw  [fill={rgb, 255:red, 255; green, 255; blue, 255 }  ,fill opacity=1 ][line width=0.75]  (355.75,213.15) -- (368.9,213.15) -- (368.9,262.5) -- (355.75,262.5) -- cycle ;
\draw  [fill={rgb, 255:red, 184; green, 233; blue, 134 }  ,fill opacity=1 ] (372.42,234.02) -- (375.03,234.06) -- (375.04,231.64) -- (378.53,236.54) -- (375,241.34) -- (375.01,238.91) -- (372.4,238.88) -- cycle ;
\draw  [fill={rgb, 255:red, 255; green, 255; blue, 255 }  ,fill opacity=1 ][line width=0.75]  (441.41,131.23) -- (454.9,131.23) -- (454.9,228.73) -- (441.41,228.73) -- cycle ;
\draw  [fill={rgb, 255:red, 155; green, 155; blue, 155 }  ,fill opacity=1 ] (459.47,177.39) -- (461.98,177.39) -- (461.98,175.18) -- (465.35,179.6) -- (461.98,184.02) -- (461.98,181.81) -- (459.47,181.81) -- cycle ;
\draw  [fill={rgb, 255:red, 184; green, 233; blue, 134 }  ,fill opacity=1 ] (427.17,205.16) -- (430.12,205.17) -- (430.12,202.74) -- (434.09,207.6) -- (430.12,212.45) -- (430.12,210.02) -- (427.17,210.02) -- cycle ;
\draw  [fill={rgb, 255:red, 74; green, 144; blue, 226 }  ,fill opacity=0.13 ] (326.38,213.15) -- (329.66,213.15) -- (329.66,262.5) -- (326.38,262.5) -- cycle ;
\draw  [fill={rgb, 255:red, 74; green, 144; blue, 226 }  ,fill opacity=0.13 ] (382.49,170.71) -- (385.77,170.71) -- (385.77,244.13) -- (382.49,244.13) -- cycle ;
\draw  [fill={rgb, 255:red, 247; green, 74; blue, 1 }  ,fill opacity=0.5 ] (55.03,318.12) -- (58.14,318.16) -- (58.15,315.57) -- (62.31,320.81) -- (58.1,325.94) -- (58.11,323.35) -- (55,323.31) -- cycle ;
\draw  [fill={rgb, 255:red, 74; green, 144; blue, 226 }  ,fill opacity=1 ] (55.22,298.69) -- (58.33,298.69) -- (58.33,296.02) -- (62.51,301.37) -- (58.33,306.72) -- (58.33,304.05) -- (55.22,304.05) -- cycle ;
\draw  [fill={rgb, 255:red, 184; green, 233; blue, 134 }  ,fill opacity=1 ] (338.52,301.14) -- (341.63,301.18) -- (341.64,298.59) -- (345.79,303.83) -- (341.58,308.96) -- (341.6,306.37) -- (338.49,306.33) -- cycle ;
\draw  [fill={rgb, 255:red, 155; green, 155; blue, 155 }  ,fill opacity=1 ] (55.63,336.28) -- (58.62,336.28) -- (58.62,333.91) -- (62.64,338.64) -- (58.62,343.37) -- (58.62,341) -- (55.63,341) -- cycle ;
\draw  [fill={rgb, 255:red, 128; green, 128; blue, 128 }  ,fill opacity=0.2 ][line width=0.75]  (468.98,131.85) -- (558.12,131.85) -- (558.12,228.34) -- (468.98,228.34) -- cycle ;
\draw  [fill={rgb, 255:red, 74; green, 144; blue, 226 }  ,fill opacity=0.13 ][line width=0.75]  (140.83,132.52) -- (153.36,132.52) -- (153.36,230.02) -- (140.83,230.02) -- cycle ;
\draw  [fill={rgb, 255:red, 255; green, 255; blue, 255 }  ,fill opacity=1 ][line width=0.75]  (165.6,170.71) -- (178.49,170.71) -- (178.49,244.13) -- (165.6,244.13) -- cycle ;
\draw  [fill={rgb, 255:red, 74; green, 144; blue, 226 }  ,fill opacity=1 ] (132.1,179.06) -- (134.6,179.06) -- (134.6,176.85) -- (137.97,181.27) -- (134.6,185.69) -- (134.6,183.48) -- (132.1,183.48) -- cycle ;
\draw  [dash pattern={on 4.5pt off 4.5pt}]  (158.9,138.52) -- (430.9,138.52) ;
\draw [shift={(433.9,138.52)}, rotate = 180] [fill={rgb, 255:red, 0; green, 0; blue, 0 }  ][line width=0.08]  [draw opacity=0] (8.93,-4.29) -- (0,0) -- (8.93,4.29) -- cycle    ;
\draw  [dash pattern={on 4.5pt off 4.5pt}]  (208.27,179.41) -- (376.6,179.41) ;
\draw [shift={(379.6,179.41)}, rotate = 180] [fill={rgb, 255:red, 0; green, 0; blue, 0 }  ][line width=0.08]  [draw opacity=0] (8.93,-4.29) -- (0,0) -- (8.93,4.29) -- cycle    ;
\draw  [fill={rgb, 255:red, 247; green, 74; blue, 1 }  ,fill opacity=0.5 ] (156.87,205.94) -- (159.48,205.98) -- (159.49,203.56) -- (162.97,208.46) -- (159.44,213.26) -- (159.46,210.83) -- (156.85,210.8) -- cycle ;
\draw  [fill={rgb, 255:red, 247; green, 74; blue, 1 }  ,fill opacity=0.5 ] (207.65,234.63) -- (210.26,234.67) -- (210.27,232.25) -- (213.76,237.15) -- (210.23,241.95) -- (210.24,239.52) -- (207.63,239.49) -- cycle ;
\draw  [fill={rgb, 255:red, 74; green, 144; blue, 226 }  ,fill opacity=0.13 ][line width=0.75]  (190.83,170.71) -- (203.49,170.71) -- (203.49,244.13) -- (190.83,244.13) -- cycle ;
\draw  [fill={rgb, 255:red, 255; green, 255; blue, 255 }  ,fill opacity=1 ][line width=0.75]  (216.34,213.15) -- (228.9,213.15) -- (228.9,262.5) -- (216.34,262.5) -- cycle ;
\draw  [fill={rgb, 255:red, 74; green, 144; blue, 226 }  ,fill opacity=0.13 ][line width=0.75]  (242.07,213.15) -- (254.9,213.15) -- (254.9,262.5) -- (242.07,262.5) -- cycle ;
\draw  [dash pattern={on 4.5pt off 4.5pt}]  (255.9,218.52) -- (322.27,218.74) ;
\draw [shift={(325.27,218.75)}, rotate = 180.19] [fill={rgb, 255:red, 0; green, 0; blue, 0 }  ][line width=0.08]  [draw opacity=0] (8.93,-4.29) -- (0,0) -- (8.93,4.29) -- cycle    ;
\draw  [fill={rgb, 255:red, 255; green, 255; blue, 255 }  ,fill opacity=1 ][line width=0.75]  (267.9,243.4) -- (282.96,243.4) -- (282.96,267.47) -- (267.9,267.47) -- cycle ;
\draw  [fill={rgb, 255:red, 247; green, 74; blue, 1 }  ,fill opacity=0.5 ] (258.92,252.32) -- (261.53,252.36) -- (261.54,249.93) -- (265.02,254.83) -- (261.49,259.64) -- (261.5,257.21) -- (258.9,257.17) -- cycle ;
\draw  [fill={rgb, 255:red, 74; green, 144; blue, 226 }  ,fill opacity=1 ] (286.78,252.81) -- (289.28,252.81) -- (289.28,250.6) -- (292.65,255.02) -- (289.28,259.44) -- (289.28,257.23) -- (286.78,257.23) -- cycle ;
\draw  [fill={rgb, 255:red, 255; green, 255; blue, 255 }  ,fill opacity=1 ][line width=0.75]  (296,242.92) -- (310.9,242.92) -- (310.9,266.99) -- (296,266.99) -- cycle ;
\draw  [fill={rgb, 255:red, 184; green, 233; blue, 134 }  ,fill opacity=1 ] (314.23,252.32) -- (316.84,252.36) -- (316.85,249.93) -- (320.33,254.83) -- (316.8,259.64) -- (316.82,257.21) -- (314.21,257.17) -- cycle ;
\draw (84.15,181.35) node  {\includegraphics[width=66.98pt,height=72.99pt]{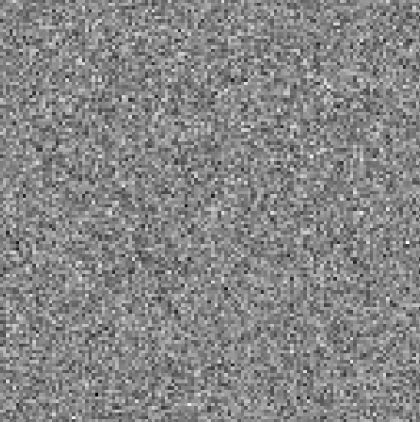}};
\draw (513.55,180.1) node  {\includegraphics[width=66.85pt,height=72.37pt]{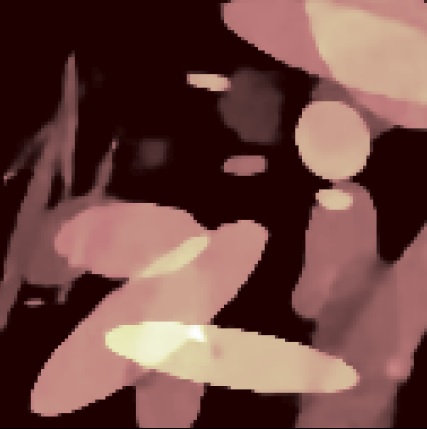}};
\draw (591.01,189.02) node  {\includegraphics[width=29.83pt,height=85.98pt]{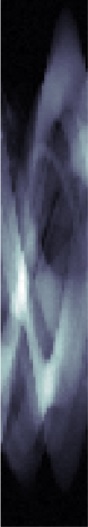}};
\draw    (536.15,128.8) -- (536.28,103.6) ;
\draw [shift={(536.3,100.6)}, rotate = 450.3] [fill={rgb, 255:red, 0; green, 0; blue, 0 }  ][line width=0.08]  [draw opacity=0] (8.93,-4.29) -- (0,0) -- (8.93,4.29) -- cycle    ;
\draw    (587.15,129.3) -- (587.28,103.6) ;
\draw [shift={(587.3,100.6)}, rotate = 450.3] [fill={rgb, 255:red, 0; green, 0; blue, 0 }  ][line width=0.08]  [draw opacity=0] (8.93,-4.29) -- (0,0) -- (8.93,4.29) -- cycle    ;
\draw  [line width=0.75]  (571.12,131.85) -- (611.27,131.85) -- (611.27,246.35) -- (571.12,246.35) -- cycle ;
\draw  [fill={rgb, 255:red, 255; green, 255; blue, 255 }  ,fill opacity=1 ][line width=0.75]  (385.6,170.71) -- (398.49,170.71) -- (398.49,244.13) -- (385.6,244.13) -- cycle ;
\draw  [fill={rgb, 255:red, 74; green, 144; blue, 226 }  ,fill opacity=1 ] (401.2,205.43) -- (403.71,205.43) -- (403.71,203.22) -- (407.08,207.64) -- (403.71,212.07) -- (403.71,209.85) -- (401.2,209.85) -- cycle ;
\draw  [fill={rgb, 255:red, 74; green, 144; blue, 226 }  ,fill opacity=0.13 ][line width=0.75]  (410.83,170.71) -- (423.49,170.71) -- (423.49,244.13) -- (410.83,244.13) -- cycle ;
\draw  [fill={rgb, 255:red, 74; green, 144; blue, 226 }  ,fill opacity=1 ] (233.15,234.63) -- (235.76,234.67) -- (235.77,232.25) -- (239.26,237.15) -- (235.73,241.95) -- (235.74,239.52) -- (233.13,239.49) -- cycle ;
\draw  [fill={rgb, 255:red, 74; green, 144; blue, 226 }  ,fill opacity=1 ] (346.42,234.02) -- (349.03,234.06) -- (349.04,231.64) -- (352.53,236.54) -- (349,241.34) -- (349.01,238.91) -- (346.4,238.88) -- cycle ;
\draw  [fill={rgb, 255:red, 74; green, 144; blue, 226 }  ,fill opacity=1 ] (181.87,205.94) -- (184.48,205.98) -- (184.49,203.56) -- (187.97,208.46) -- (184.44,213.26) -- (184.46,210.83) -- (181.85,210.8) -- cycle ;

\draw (90.32,118.83) node  [font=\large]  {$z$};
\draw (164.11,301.14) node  [font=\footnotesize] [align=left] {$\displaystyle 3\times 3$ Conv + Bn + LeakyRelu};
\draw (185.02,320.43) node  [font=\footnotesize] [align=left] {$\displaystyle 3\times 3$ Stride-Conv + Bn + LeakyRelu};
\draw (480.57,302.61) node  [font=\footnotesize] [align=left] {Upsample + $\displaystyle 3\times 3$ Conv + Bn + LeakyRelu};
\draw (104.88,338.58) node  [font=\footnotesize] [align=left] {$\displaystyle 1\times 1$ conv};
\draw (333.34,204.44) node  [font=\footnotesize] [align=left] {$\displaystyle 128$};
\draw (362.84,204.27) node  [font=\footnotesize] [align=left] {$\displaystyle 128$};
\draw (431.43,320.17) node  [font=\footnotesize] [align=left] {$\displaystyle 1\times 1$ conv + concatenation};
\draw (146.85,123.38) node  [font=\footnotesize] [align=left] {$\displaystyle 128$};
\draw (200.01,161.72) node  [font=\footnotesize] [align=left] {$\displaystyle 128$};
\draw (173.18,161.72) node  [font=\footnotesize] [align=left] {$\displaystyle 128$};
\draw (222.72,204.43) node  [font=\footnotesize] [align=left] {$\displaystyle 128$};
\draw (248.75,204.27) node  [font=\footnotesize] [align=left] {$\displaystyle 128$};
\draw (273.93,235.35) node  [font=\footnotesize] [align=left] {$\displaystyle 128$};
\draw (302.34,234.52) node  [font=\footnotesize] [align=left] {$\displaystyle 128$};
\draw (446.78,120.88) node  [font=\footnotesize] [align=left] {$\displaystyle 128$};
\draw (532.84,72.56) node  [font=\normalsize]  {${\displaystyle {{\mathrm{min\ }}}\frac{1}{2}\left\Vert A\underbrace{{\displaystyle \varphi(\theta, z)}} -\underbrace{{\displaystyle y^{\delta }}}\right\Vert ^{2}}$};
\draw (464.32,88.83) node  [font=\normalsize]  {$\theta $};
\draw (418.01,161.72) node  [font=\footnotesize] [align=left] {$\displaystyle 128$};
\draw (391.18,161.72) node  [font=\footnotesize] [align=left] {$\displaystyle 128$};

\end{tikzpicture}

%% file: table_ellipses_data_sizes.tex
\begin{tabular}{lrrrrrrrrrr}
\% &\num{0.1}&\num{0.2}&\num{0.5}&\num{1.0}&\num{2.0}&\num{5.0}&\num{10.0}&\num{25.0}&\num{50.0}&\num{100.0}\\\toprule
\#train &\num{32}&\num{64}&\num{160}&\num{320}&\num{640}&\num{1600}&\num{3200}&\num{8000}&\num{16000}&\num{32000}\\
\#val &\num{3}&\num{6}&\num{16}&\num{32}&\num{64}&\num{160}&\num{320}&\num{800}&\num{1600}&\num{3200}\\\bottomrule
\end{tabular}

%% file: table_lodopab_data_sizes.tex
\begin{tabular}{lrrrrr}
\% &\num{0.01}&\num{0.1}&\num{1.0}&\num{10.0}&\num{100.0}\\\toprule
\#train &\num{3}&\num{35}&\num{358}&\num{3582}&\num{35820}\\
\#val &\num{1}&\num{3}&\num{35}&\num{352}&\num{3522}\\\midrule
\#patients train &\num{1}&\num{1}&\num{7}&\num{64}&\num{632}\\
\#patients val &\num{1}&\num{1}&\num{1}&\num{6}&\num{60}\\\bottomrule
\end{tabular}